\documentclass[reqno,tbtags,intlimits,a4paper,oneside,12pt]{amsart}
\usepackage[cp1251]{inputenc}%
\usepackage[english]{babel}
\usepackage{amssymb,upref,calc,url}
\usepackage[mathscr]{eucal}
\usepackage{graphicx}
\usepackage{amsmath,amscd,amsthm,verbatim}
\usepackage[all]{xy}
\usepackage[in]{fullpage}
\newtheorem{thm}{\hskip\parindent Theorem}[section]
\newtheorem{lem}[thm]{\hskip\parindent Lemma}
\newtheorem{prb}[thm]{\hskip\parindent Problem}
\newtheorem{cor}[thm]{\hskip\parindent Corollary}
\newtheorem{rem}[thm]{\hskip\parindent Remark}
\newtheorem{ex}[thm]{\hskip\parindent Example}

\theoremstyle{definition}

\DeclareMathOperator{\Der}{Der}
\DeclareMathOperator{\wt}{wt}

\begin{document}

\title{Polynomial dynamical systems and~differentiation of~genus~$4$~hyperelliptic~functions}
\author{E.\,Yu.~Bunkova}
\address{Steklov Mathematical Institute of Russian Academy of Sciences, Moscow, Russia}
\email{buchstab@mi-ras.ru, bunkova@mi-ras.ru}

\begin{abstract}
We give an explicit solution to the problem of differentiation of hyperelliptic functions in genus $4$ case. We describe explicitly the polynomial Lie algebras and polynomial dynamical systems connected to this problem.
\end{abstract}
\maketitle

\vspace{-15pt}

\section{Introduction} \label{S0}

Let $g \in \mathbb{N}$. We denote the coordinates in~the complex space $\mathbb{C}^g$ by~$t = (t_1, t_3, \ldots, t_{2g-1})$. For a meromorphic function $f$ on $\mathbb{C}^g$, a vector $\omega \in \mathbb{C}^g$ is~a \emph{period} if $f(t+\omega) = f(t)$ for all~$t \in \mathbb{C}^g$.
If a meromorphic function $f$
has $2g$ independent periods in~$\mathbb{C}^g$, then $f$ is~called an \emph{Abelian function}.
Thus, an Abelian function is a meromorphic~function on the complex torus $\mathbb{C}^g/\Gamma$,
where $\Gamma$ is the lattice formed by the periods. See \cite{BMulti}.

A plane nonsingular algebraic curve of genus $g$ determines a lattice $\Gamma$ as the set of periods of its holomorphic differentials. The torus $\mathbb{C}^g/\Gamma$ is the \emph{Jacobian variety} of the~curve.

In \cite{BL} the problem of differentiation of the field of Abelian functions on Jacobian varieties of genus $g$ curves is considered. In this work we consider a special case of this problem. Namely, we consider the model of
the \emph{universal hyperelliptic curve of genus} $g$
\[
\mathcal{V}_\lambda = \{(x, y)\in\mathbb{C}^2 \colon
y^2 = x^{2g+1} + \lambda_4 x^{2 g - 1}  + \lambda_6 x^{2 g - 2} + \ldots + \lambda_{4 g} x + \lambda_{4 g + 2}\}. 
\]
Each curve is defined by specialization of parameters $\lambda = (\lambda_4, \lambda_6, \ldots, \lambda_{4 g}, \lambda_{4 g + 2}) \in \mathbb{C}^{2 g}$.

The indices of all the coordinates $t = (t_1, t_3, \ldots, t_{2g-1}) \in \mathbb{C}^g$ and of the parameters $\lambda = (\lambda_4, \lambda_6, \ldots, \lambda_{4 g}, \lambda_{4 g + 2}) \in \mathbb{C}^{2 g}$ determine their weights. Namely,
$\wt t_k = - k$ and~$\wt \lambda_k = k$.
For~suitable weights of the~other variables, all the equations in this paper are of~homogeneous weight.

Let $\mathcal{B} \subset \mathbb{C}^{2g}$ be the subspace of parameters such that the curve~$\mathcal{V}_{\lambda}$ is nonsingular for~$\lambda \in~\mathcal{B}$.
Then we have $\mathcal{B} = \mathbb{C}^{2g} \backslash \Sigma$, where $\Sigma$ is~the \emph{discriminant hypersurface} of~the universal curve. In Section \ref{s2} we describe the polynomial vector fields in $\mathcal{B}$ tangent to $\Sigma$.

For each $\lambda \in \mathcal{B}$ the set of periods of holomorphic differentials on the curve $\mathcal{V}_\lambda$ generates a lattice $\Gamma_\lambda$ of rank $2 g$ in $\mathbb{C}^g$.
A \emph{hyperelliptic function of genus}~$g$ is~a~meromorphic function on $\mathbb{C}^g \times \mathcal{B}$
such that, for each $\lambda \in \mathcal{B}$, its restriction to~\text{$\mathbb{C}^g\times\lambda$}
is~an~Abelian function with lattice of periods $\Gamma_\lambda$.
Thus, a hyperelliptic function is a~function defined on an open dense subset of the total space $\mathcal{U}$ of the fiber bundle $\pi\colon \mathcal{U} \to \mathcal{B}$ with fiber over $\lambda \in \mathcal{B}$ the~Jacobian variety $\mathcal{J}_\lambda = \mathbb{C}^g/\Gamma_\lambda$ of the curve~$\mathcal{V}_\lambda$. The universal bundle of Jacobians of hyperelliptic curves $\mathcal{U}$ is introduced in~\cite{DN}.
We~denote by~$\mathcal{F}$ the~field of hyperelliptic functions of genus $g$. 

The \emph{problem of differentiation of hyperelliptic functions} is a genus $g$ analogue of a~genus~$1$ result of \cite{FS}, see also \S 1.2 in \cite{BL}. We consider this problem in the form:
\begin{prb}[Problem 1.1 in \cite{B3}] \text{}  \label{p1}
\begin{enumerate} 
 \item Find the $3g$ generators of the $\mathcal{F}$-module $\Der \mathcal{F}$ of derivations of the field $\mathcal{F}$.
 \item Describe the structure of Lie algebra $\Der \mathcal{F}$ (i.e. find the commutation relations).
\end{enumerate}
\end{prb} 

A general approach to the solution of this problem is given in \cite{BL}. An~overview with examples in the cases of genus $g=1$ elliptic and $g=2$ hyperelliptic curves is given in \cite{BEL18}. 

We use the theory of hyperelliptic Kleinian functions (see \cite{Baker}, \cite{BEL}, \cite{BEL-97}, \cite{BEL-12}, and~\cite{WW} for elliptic functions).
Take coordinates
$(t, \lambda)$
in $\mathbb{C}^g \times \mathcal{B} \subset \mathbb{C}^{3g}$.
Let $\sigma(t, \lambda)$ be the hyperelliptic sigma function. We~denote~$\partial_k = {\partial \over \partial t_k}$.
We use the notation
\[
\zeta_{k} = \partial_k \ln \sigma(t, \lambda), \qquad
\wp_{k_1, \ldots, k_n} = - \partial_{k_1} \cdots \partial_{k_n} \ln \sigma(t, \lambda),
\]
where $n \geqslant 2$ and $k_s \in \{ 1, 3, \ldots, 2 g - 1\}$.
The functions $\wp_{k_1, \ldots, k_n}$ give examples of~hyperelliptic functions. 

Consider the diagram \vspace{-16pt}
\begin{equation} \label{d}
 \xymatrix{
	\mathcal{U} \ar[d]^{\pi} \ar@{-->}[r]^{\varphi} & \mathbb{C}^{3 g} \ar[d]^{\rho}\\
	\mathcal{B} \ar@{^{(}->}[r] & \mathbb{C}^{2g}\\
	}  
\end{equation}
The fiber bundle $\pi: \mathcal{U} \to \mathcal{B}$ and the embedding $\mathcal{B} \subset \mathbb{C}^{2g}$ are described above.
In~Section~\ref{s3} we describe the maps $\varphi$ and $\rho$ following \cite{B3}. 
By Dubrovin--Novikov theorem~\cite{DN}, the space~$\mathcal{U}$ is birationally isomorphic to the complex linear space $\mathbb{C}^{3g}$. We construct such an~isomorphism $\varphi$ explicitly.
We use a~fundamental result from the theory of hyperelliptic Abelian functions (see~\cite{BEL-12}, Chapter~5):
Any hyperelliptic function can be represented as a rational function
in $\wp_{1,k}$ and $\wp_{1,1,k}$, where $k \in \{1, 3, \ldots, 2 g - 1\}$.
Theorem \ref{t31} in Section \ref{s3} gives a set of relations between the derivatives of these functions.
We use it to introduce a~set of generators in~$\mathcal{F}$.
The map $\varphi$ will be determined by this set of generators. The~map~$\rho$ will be a polynomial map that makes the diagram~\eqref{d} commutative.

We denote the ring of polynomials in $\lambda \in \mathbb{C}^{2g}$ by $\mathcal{P}$.
For the polynomial map~$\rho\colon \mathbb{C}^{3g} \to \mathbb{C}^{2g}$ we call a vector field $\mathcal{D}$ in $\mathbb{C}^{3g}$ \emph{projectable} if there exists a vector field $L$ in $\mathbb{C}^{2g}$ such that
$\mathcal{D}(\rho^* f) = \rho^* L(f)$ for any $f \in \mathcal{P}$.
The vector field $L$ is the \emph{pushforward} of $\mathcal{D}$.
A corollary of this definition is that for a~projectable vector field $\mathcal{D}$ we have $\mathcal{D}(\rho^* \mathcal{P}) \subset \rho^* \mathcal{P}$.

We consider the problem:

\begin{prb} \label{p2}
Find $3 g$ polynomial vector fields in $\mathbb{C}^{3g}$ projectable for $\rho\colon \mathbb{C}^{3g} \to \mathbb{C}^{2g}$ and
independent at any point in $\rho^{-1}(\mathcal{B})$.
Construct their polynomial Lie algebra.
\end{prb}

There is a direct relation \eqref{reld} between Problems 1.1 and 1.2. In Section \ref{s4} we give vector fields that give part of both solutions for general genus $g$. The pushforwards of~these solutions are the polynomial vector fields described in Section \ref{s2}.
In Section \ref{Sdyn} we give polynomial dynamical systems that arise from this relation.

In the case of genus~$g=2$ the solution to Problem \ref{p2} and the corresponding homogeneous polynomial dynamical systems in $\mathbb{C}^6$ are given in Section 4 of \cite{B2}. An explicit solution to Problem \ref{p1} is given in Theorems B.3 and B.6 of \cite{B2}. We recall these results in~Section~\ref{S2L}.

In the case of genus~$g=3$ an explicit solution to Problem \ref{p1} is given in \cite{B3}. The~proof  is~based on a solution to Problem \ref{p2}. We give these results in~Section~\ref{S3L}. The corresponding polynomial dynamical systems in $\mathbb{C}^9$ are presented in Section 8 of \cite{BB21}.

The case of genus~$g=4$ is considered in Sections \ref{S4L}, \ref{s10}, \ref{scom4}, and \ref{S4Sys} of this work. We~find the solution to Problem \ref{p1} in Theorem \ref{tg4} and the solution to Problem \ref{p2} in~Theorem~\ref{td4}. Some of the results of these Sections were first obtained in \cite{BArch}. Here we complete the~descriptions of the Lie algebras, see Section \ref{scom4}. Our approach is based on a construction of the generators of $\Der \mathcal{F}$ based on a result from \cite{BB20}.

In all the cases that we consider the generators $\mathcal{L}_k$ of $\Der \mathcal{F}$, that solve Problem \ref{p1}, and the polynomial vector fields $\mathcal{D}_k$ in $\mathbb{C}^{3g}$, that solve Problem \ref{p2}, are related as
\begin{equation} \label{reld}
\mathcal{L}_k (\varphi^* b_{i,j}) = \varphi^* \mathcal{D}_k (b_{i,j}) 
\end{equation}
for coordinate functions $b_{i,j} \in \mathbb{C}^{3g}$. We note that these solutions give explicit solutions of~the Problem considered in \cite{BL} for hyperelliptic curves $\mathcal{V}_\lambda$ of genus $2$, $3$, and $4$.

\vfill
\eject

\section{Polynomial vector fields tangent to the discriminant hypersurface of~the universal curve} \label{s2}

Let us define the \emph{polynomial Lie algebra of vector fields tangent to the discriminant hypersurface} $\Sigma$ \emph{in}~$\mathbb{C}^{2g}$.
We denote it by~$\mathscr{L}_{\mathcal{B}}$.
Here $\mathbb{C}^{2g}$ is the complex linear space with coordinates $(\lambda)$, $\mathcal{P}$ is the ring of polynomials in $(\lambda)$ and $\Sigma$ is defined in Section \ref{S0}.

For the polynomial Lie algebra $\mathscr{L}_{\mathcal{B}}$,
the generators $\{L_{0}, L_{2}, L_{4}, \ldots, L_{4 g - 2}\}$ are the vector fields
\[
 L_{2k} = \sum_{s=2}^{2g+1} v_{2k+2, 2s-2}(\lambda) {\partial \over \partial \lambda_{2k}}, \qquad v_{2k+2, 2s-2}(\lambda) \in \mathcal{P}.
\]

By~\cite{BPol}, the~structure of a  polynomial Lie algebra as a $\mathcal{P}$-module with
generators
$
1, \; L_{0}, \; L_{2}, \; L_{4}, \; \ldots, \; L_{4 g - 2}
$ is determined by the polynomial matrices $V(\lambda)=(v_{2i,2j}(\lambda)),$ where
$i,j=1,\dots,2g$, and~$C(\lambda) = (c_{2i,2j}^{2k}(\lambda))$, where $i,j,k=0,\dots,2g-1$, such that 
\begin{equation} \label{fcijk}
[L_{2i},L_{2j}]=\sum_{k=0}^{2g-1} c_{2i,2j}^{2k}(\lambda)L_{2k},
\quad[L_{2i},\lambda_{2j+4}]=v_{2i+2,2j+2}(\lambda),
\quad[\lambda_{2i+4},\lambda_{2j+4}]=0. 
\end{equation}

In the case of the Lie algebra $\mathscr{L}_{\mathcal{B}}$,
explicit expressions for the matrix $V(\lambda)$ 
can be~found in Section 4.1 of \cite{A}. 
The elements of this matrix are given by the following formulas. For~convenience, we~assume that $\lambda_s = 0$ for all $s \notin \{0,4,6, \ldots, 4 g, 4 g + 2\}$ and~$\lambda_0 = 1$.
Let $k, m \in \{ 1, 2, \ldots, 2 g\}$. If $k\leqslant m$, then we~set
\begin{equation} \label{ev}
v_{2k, 2m}(\lambda) = \sum_{s=0}^{k-1} 2 (k + m - 2 s) \lambda_{2s} \lambda_{2 (k+m-s)}
 - {2 k (2 g - m + 1) \over 2 g + 1} \lambda_{2k} \lambda_{2m}, 
\end{equation}
and if $k > m $, then we set $v_{2k, 2m}(\lambda) = v_{2m, 2k}(\lambda)$.

The vector field $L_0$ is the Euler vector field; namely, since $\wt \lambda_{2k} = 2 k$, we have
\begin{align*} 
 [L_0, \lambda_{2k}] &= 2 k \lambda_{2k}, & [L_0, L_{2k}] &= 2 k L_{2k}.
\end{align*}
This determines the weights of the vector fields $L_k$, namely, $\wt L_{2k} = 2 k$. 

\begin{lem}[Lemma 4.3 in \cite{BB20}] \label{lcijk} \text{} 
\[
[L_2, L_{2k}] = 2 (k-1) L_{2k+2} + {4 (2 g - k) \over (2 g + 1)} \left( \lambda_{2k+2} L_0 - \lambda_4 L_{2k-2}\right).  
\]
\end{lem}
This Lemma determines the polynomials $c_{2,2j}^{2k}(\lambda)$ in \eqref{fcijk}. For general $i$ the polynomials~$c_{2i,2j}^{2k}(\lambda)$ in \eqref{fcijk} are described in Theorem 2.5 of \cite{BL2004}. We can derive them directly from the explicit expressions \eqref{ev}.

\begin{ex}
In the case of genus~$g=2$ we have
\begin{align*}
[L_2, L_4] &=\frac{8}{5}\lambda_6 L_0 -\frac{8}{5}\lambda_4 L_2 + 2 L_6, \qquad
[L_2, L_6] = \frac{4}{5}\lambda_8 L_0 -\frac{4}{5}\lambda_4 L_4, & \\
[L_4, L_6] &=-2\lambda_{10} L_0 +\frac{6}{5}\lambda_8 L_2 -\frac{6}{5}\lambda_6 L_4 + 2\lambda_4 L_6. 
\end{align*}
\end{ex}

\begin{ex}
In the case of genus~$g=3$ 
the expressions for $c_{2i,2j}^{2k}$ are given in Lemma~4.3 of \cite{B3}. In particular, we have:
\begin{align} \label{3ex}
 [L_2, L_{4}] &= 2 L_{6} + {16 \over 7} \left( \lambda_{6} L_0 - \lambda_4 L_{2}\right), \qquad
[L_2, L_{6}] = 4 L_{8} + {12 \over 7} \left( \lambda_{8} L_0 - \lambda_4 L_{4}\right),
\\
[L_2, L_{8}] &= 6 L_{10} + {8 \over 7} \left( \lambda_{10} L_0 - \lambda_4 L_{6}\right). \nonumber
\end{align}
\end{ex}

In the case of genus~$g=4$ 
the expressions for $c_{2,2j}^{2k}$ are given in
Example \ref{ex61}.

\vfill

\eject

\section{Polynomial map related to the~universal~bundle~of~Jacobians~of~hyperelliptic~curves} \label{s3}

\begin{thm}[\cite{BEL}, \S 3 in \cite{BEL-97}] \label{t31} For $i, k \in \{1, 3, \ldots, 2 g - 1\}$ we~have the relations
\begin{align*}
\wp_{1,1,1,i} = & 6 \wp_{1,1} \wp_{1,i} + 6 \wp_{1, i+2} - 2 \wp_{3, i} + 2 \lambda_{4} \delta_{i,1},\\
\wp_{1,1,i} \wp_{1,1,k} = & 4 \left(\wp_{1,1} \wp_{1,i} \wp_{1,k} + \wp_{1,k} \wp_{1,i+2} + \wp_{1,i} \wp_{1, k+2} + \wp_{k+2, i+2}\right)  
   - \\
   & - 2 (\wp_{1,i} \wp_{3,k} + \wp_{1,k} \wp_{3,i} + \wp_{k, i+4} + \wp_{i, k+4}) + \\
   &+ 2 \lambda_4 (\delta_{i,1} \wp_{1,k} + \delta_{k,1} \wp_{1,i}) + 2 \lambda_{i+k+4} (2 \delta_{i,k} + \delta_{k, i-2} + \delta_{i, k-2}).
\end{align*}
\end{thm}

\begin{cor}[Corollary 5.2 in \cite{B3}] \label{c1}
 Consider the map $\widetilde{\varphi}: \mathcal{U} \dashrightarrow \mathbb{C}^{ \frac{g (g+9)}{2}}$ with the coordinates $(b, p, \lambda)$ in $\mathbb{C}^{ \frac{g (g+9)}{2}}$. Here $b = (b_{i,j}) \in \mathbb{C}^{3g}$ with $i \in \{ 1,2,3 \}$, \mbox{$j \in \{1, 3, \ldots, 2 g -1\}$,}
 $p = (p_{k,l}) \in \mathbb{C}^{\frac{g (g-1)}{2}}$ with $k,l \in \{3, 5, \ldots, 2 g -1\}$ for $k \leqslant l$,
 and $\lambda = (\lambda_s) \in \mathbb{C}^{2 g}$ with~$s \in \{4, 6, \ldots, 4 g, 4 g + 2\}$.
 Set 
 \[
  \widetilde{\varphi}: (t, \lambda) \mapsto (b_{1,j}, b_{2,j}, b_{3,j}, p_{k,l}, \lambda_s) = (\wp_{1,j}(t, \lambda), \wp_{1,1,j}(t, \lambda), \wp_{1,1,1,j}(t, \lambda), 2 \wp_{k,l}(t, \lambda), \lambda_s).
 \]
 We denote $p_{l,k} = p_{k,l}$ for $k,l \in \{3, 5, \ldots, 2 g -1\}$.
Then the image of $\widetilde{\varphi}$ lies in $\mathcal{S} \subset \mathbb{C}^{ \frac{g (g+9)}{2}}$, where $\mathcal{S}$ is determined by the set of $g (g+3)/2$ equations
\begin{align*}
b_{3,1} = & 6 b_{1,1}^2 + 4 b_{1,3} + 2 \lambda_{4},\\
b_{3, k} = & 6 b_{1,1} b_{1,k} + 6 b_{1,k+2} - p_{3, k}, \\
b_{2,1}^2 = & 4 b_{1,1}^3 + 4 b_{1,1} b_{1, 3} - 4 b_{1, 5} + 2 p_{3, 3} + 4 \lambda_4 b_{1,1} + 4 \lambda_{6}, \\
b_{2,1} b_{2, k} = & 4 b_{1,1}^2 b_{1, k} + 2 b_{1, 3} b_{1, k} + 4 b_{1,1} b_{1, k+2} - 2 b_{1, k+4} - \\
   & - b_{1,1} p_{3,k} + 2 p_{3, k+2} - p_{5, k} + 2 \lambda_4 b_{1,k} + 2 \lambda_{8} \delta_{3,k}, \\
b_{2,j} b_{2,k} = & 4 b_{1,1} b_{1,j} b_{1,k} + 4 b_{1,k} b_{1,j+2} + 4 b_{1,j} b_{1,k+2} + 2 p_{k+2, j+2}
   - \\
   & - b_{1,j} p_{3,k} - b_{1,k} p_{3,j} - p_{k, j+4} - p_{j, k+4} + 2 \lambda_{j+k+4} (2 \delta_{j,k} + \delta_{k, j-2} + \delta_{j, k-2}),
\end{align*}
for $j,k \in \{3, \ldots, 2g-1\}$ and any variable equal to zero if the index is out of range.
\end{cor}

\begin{thm}[Theorem 5.3 in \cite{B3}] \label{thm3}
The projection $\pi_1\colon \mathbb{C}^{\frac{g (g+9)}{2}} \to \mathbb{C}^{3 g}$ on the first~$3 g$ coordinates gives the isomorphism $\mathcal{S} \simeq \mathbb{C}^{3g}$.
Therefore, the coordinates $(b_{i,j})$ uniformize~$\mathcal{S}$.
\end{thm}

We have $\widetilde{\varphi}: \mathcal{U} \dashrightarrow \mathcal{S} \simeq \mathbb{C}^{3g}$ and we denote by $\varphi$ the composition $\pi_1 \circ \widetilde{\varphi}: \mathcal{U} \dashrightarrow \mathbb{C}^{3g}$.

We obtain the diagram \vspace{-17pt}
\begin{equation} \label{diagbig}
\xymatrix{
	 & \mathbb{C}^{\frac{g (g+9)}{2}} \ar@{<-_{)}}[d] \ar@{=}[r] & \mathbb{C}^{3 g} \times \mathbb{C}^{\frac{g (g-1)}{2}} \times \mathbb{C}^{2 g} \ar[dl]^{\pi_1} \ar@/^/[ddl]^{\pi_3}\\
	\mathcal{U} \ar[d]^{\pi} \ar@{-->}[r]^(.4){\varphi} \ar@{-->}[ur]^(.5){\widetilde{\varphi}}& \mathcal{S} \simeq \mathbb{C}^{3 g} \ar[d]^{\rho}\\
	\mathcal{B} \ar@{^{(}->}[r] & \mathbb{C}^{2g}
	}
\end{equation}

\begin{cor}[Corollary 5.5 in \cite{B3}] \label{cor3}
The projection $\pi_3\colon \mathbb{C}^{\frac{g (g+9)}{2}} \to \mathbb{C}^{2 g}$ on the last $2 g$ coordinates in Corollary \ref{c1} restricted to $\mathcal{S} \simeq \mathbb{C}^{3 g}$ gives the polynomial map $\rho\colon \mathbb{C}^{3g} \to \mathbb{C}^{2g}$.
\end{cor}

\begin{cor}[Corollary 5.4 in \cite{B3}] \label{cor4}
The projection $\pi_2\colon \mathbb{C}^{\frac{g (g+9)}{2}} \to  \mathbb{C}^{\frac{g (g-1)}{2}}$ on the middle $g (g-1)/2$ coordinates in Corollary \ref{c1} restricted to $\mathcal{S} \simeq \mathbb{C}^{3 g}$ gives a polynomial map $\mathbb{C}^{3g} \to \mathbb{C}^{\frac{g (g-1)}{2}}$.
\end{cor}

The proof of these Corollaries in \cite{B3} allows to obtain recursive expressions for the polynomials $\lambda_s$ and $p_{k,l}$ in $b_{i,j}$.
We will describe the polynomial maps $\rho\colon \mathbb{C}^{3g} \to \mathbb{C}^{2g}$ and~$\mathbb{C}^{3g} \to \mathbb{C}^{\frac{g (g-1)}{2}}$ explicitly for any $g$ in our next work.

\section{Explicit description of polynomial vector fields projectable for $\rho$} \label{s4}

Now let us describe explicitly some of the polynomial vector fields $\mathcal{D}_k$ projectable for~$\rho\colon \mathbb{C}^{3g} \to \mathbb{C}^{2g}$. We construct them using the relation
\begin{equation} \label{keyrel}
\mathcal{L}_k (\varphi^* b_{i,j}) = \varphi^* \mathcal{D}_k (b_{i,j}) 
\end{equation}
for $\mathcal{L}_k \in \Der \mathcal{F}$.

\begin{rem}[Section 2 in \cite{B3}] \label{rem45}
The operators $\mathcal{L}_{2k-1} = \partial_{2k-1}$ for $k \in \{1,2,3,\ldots,g\}$ belong to the Lie algebra of derivations of $\mathcal{F}$. Their pushforwards for $\pi$ are zero.
\end{rem}

\begin{lem}[Lemma 6.2 in \cite{B3}] \label{ld1} For the polynomial vector field 
\begin{align*}
\mathcal{D}_1 &= \sum_j b_{2,j} {\partial \over \partial b_{1,j}} + b_{3, j} {\partial \over \partial b_{2,j}} + 4 (2 b_{1,1} b_{2,j} + b_{2,1} b_{1, j} + b_{2, j+2}) {\partial \over \partial b_{3,j}}
\end{align*}
where $b_{2,2 g + 1} = 0$ 
we have $\mathcal{L}_1 (\varphi^* b_{i,j}) = \varphi^* \mathcal{D}_1 (b_{i,j})$ for all $i \in \{1,2,3\}, j \in \{1,3,\ldots,2g-1\}$.
\end{lem}

\begin{lem}[cf. Lemma 6.3 in \cite{B3}] \label{ld2} Set $p_{s,1} = 2 b_{1,s}$. For $s \in \{3, 5, \ldots, 2g - 1\}$ for the polynomial vector fields
\begin{align*}
\mathcal{D}_{s} &= 
{1 \over 2} \sum_{k=1}^{g} \left( \mathcal{D}_1(p_{s,2k-1}) {\partial \over \partial b_{1,2k-1}} +  \mathcal{D}_1(\mathcal{D}_1(p_{s,2k-1})) {\partial \over \partial b_{2,2k-1}}
+  \mathcal{D}_1(\mathcal{D}_1(\mathcal{D}_1(p_{s,2k-1}))) {\partial \over \partial b_{3,2k-1}} \right)
\end{align*}
we have $\mathcal{L}_s (\varphi^* b_{i,j}) = \varphi^* \mathcal{D}_s (b_{i,j})$ for all $i \in \{1,2,3\}, j \in \{1,3,\ldots,2g-1\}$. 
\end{lem}

\begin{cor} \label{cords}
 The polynomial vector fields $\mathcal{D}_s$ for $s \in \{1,3,\ldots,2g-1\}$ are projectable for the polynomial map $\rho\colon \mathbb{C}^{3g} \to \mathbb{C}^{2g}$. Their pushforwards are zero.
\end{cor}

\begin{thm}[Theorem 6.1 in \cite{BB21}]
The operators 
\begin{align*}
\hspace{-20pt} \mathcal{L}_0 &= L_{0} - \sum_{s=1}^g (2s-1) t_{2s-1} \partial_{2s-1},  \\
\hspace{-20pt} \mathcal{L}_2 &= L_{2} - \zeta_1 \partial_1 - \sum_{s=1}^{g-1} (2s-1) t_{2s-1} \partial_{2s+1}
+ {4 \over 2 g + 1} \lambda_4  \sum_{s=1}^{g-1} (g - s)  t_{2s+1} \partial_{2s-1}, \\
\hspace{-20pt} \mathcal{L}_4 &= L_{4} - \zeta_3 \partial_1 - \zeta_1 \partial_3 - \sum_{s=1}^{g-2} (2s-1) t_{2s-1} \partial_{2s+3}
 - \\
 & \qquad \qquad \quad - \lambda_4 \sum_{s=1}^{g-1} (2s-1) t_{2s+1} \partial_{2s+1} + {6 \over 2 g + 1} \lambda_6 \sum_{s=1}^{g-1} (g - s) t_{2s+1} \partial_{2s-1} 
\end{align*}
belong to the Lie algebra of derivations of $\mathcal{F}$.
\end{thm}

\begin{lem}[Equation (22) in \cite{B3}] \label{ld0} For the polynomial vector field 
\begin{align*}
\mathcal{D}_0 &= \sum_j (j+1) b_{1,j} {\partial \over \partial b_{1,j}} + (j+2) b_{2,j} {\partial \over \partial b_{2,j}} +
(j+3) b_{3, j} {\partial \over \partial b_{3,j}}.
\end{align*}
we have $\mathcal{L}_0 (\varphi^* b_{i,j}) = \varphi^* \mathcal{D}_0 (b_{i,j})$ for all $i \in \{1,2,3\}$, $j \in \{1,3,\ldots,2g-1\}$.
\end{lem}

\begin{cor} \label{cor0}
 The polynomial vector field $\mathcal{D}_0$ is projectable for the polynomial map~$\rho$ with pushforward $L_0$.
\end{cor}

\begin{prb}
 Describe explicitly the polynomial vector fields $\mathcal{D}_2$ and $\mathcal{D}_4$ such that $\mathcal{L}_k (\varphi^* b_{i,j}) = \varphi^* \mathcal{D}_k (b_{i,j})$ for $k = 2,4$ and $i \in \{1,2,3\}, j \in \{1,3,\ldots,2g-1\}$.
\end{prb}

We will give a solution of this problem based on results of \cite{BB21} and \cite{BL} in our next work.

\vfill
\eject

\section{Polynomial dynamical systems related to~differentiations of~hyperelliptic~functions} \label{Sdyn}

In Section \ref{s4} the generators $\mathcal{L}_k$ of $\Der \mathcal{F}$ and the polynomial vector fields $\mathcal{D}_k$ in $\mathbb{C}^{3g}$ are~related by
\eqref{keyrel}. In this Section we give the
graded homogeneous polynomial dynamical systems in $\mathbb{C}^{3g}$ determined by these vector fields. We follow the approach of \cite{B2}, where such a description is given in the case of genus~$g=1$ and $g=2$. See also Section 8 of~\cite{BB21}.
By~definition, the dynamical system $S_{k}$ corresponding to the vector field $\mathcal{D}_{k}$ is given by
\[
{\partial \over \partial \tau_{k}} b_{i,j} = \mathcal{D}_{k}(b_{i,j}).
\]

The dynamical system $S_0$ corresponding to the Euler vector field $\mathcal{D}_0$ is given by
\[
 {\partial \over \partial \tau_0} b_{i,j} = (i+j) b_{i,j}, \quad \text{for} \quad i \in \{1,2,3\}, \quad j \in \{1,3,\ldots,2g-1\}.
\]
The dynamical system $S_1$ corresponding to the vector field $\mathcal{D}_1$ for $j \in \{1,3,\ldots,2g-1\}$ is given by
\begin{align*}
{\partial \over \partial \tau_1} b_{i,j} &= b_{i+1,j}, \quad \text{for} \quad i \in \{1,2\},&
{\partial \over \partial \tau_1} b_{3,j} &= 4 (2 b_{1,1} b_{2,j} + b_{2,1} b_{1, j} + b_{2, j+2}),
\end{align*}
where $b_{2,2 g + 1} = 0$. 

The dynamical systems $S_s$ for $s \in \{3, 5, \ldots, 2g - 1\}$ corresponding to vector fields $\mathcal{D}_s$ for $j \in \{1,3,\ldots,2g-1\}$ are given by
\begin{align*}
{\partial \over \partial \tau_s} b_{1,j} &= {1 \over 2} \mathcal{D}_1(p_{s,j}), &
{\partial \over \partial \tau_s} b_{2,j} &= {1 \over 2} \mathcal{D}_1(\mathcal{D}_1(p_{s,j})), &
{\partial \over \partial \tau_s} b_{3,j} &= {1 \over 2} \mathcal{D}_1(\mathcal{D}_1(\mathcal{D}_1(p_{s,j}))),
\end{align*}
where $p_{s,1} = 2 b_{1,s}$.

\section{Lie algebra of derivations of hyperelliptic functions: genus~$2$ case} \label{S2L}

The following results were obtained in \cite{B2}. We give them in the notation of this work using the explicit formulas from Section \ref{s4}.

\begin{thm}[Theorem B.3 and Theorem B.6 in~\cite{B2}] \label{tg2} In the case of genus $g=2$ the~following vector fields give a~solution to Problem \ref{p1}
\begin{align*}
\mathcal{L}_1 &= \partial_{1}, & \mathcal{L}_0 &= L_0 - t_1 \partial_{1} - 3 t_3 \partial_{3}, &
\mathcal{L}_4 &= L_4 - \zeta_3 \partial_{1} - \zeta_1 \partial_{3} - \lambda_4 t_3 \partial_{3} + {6 \over 5} \lambda_6 t_3 \partial_{1}, \\
\mathcal{L}_3 &= \partial_{3}, & \mathcal{L}_2 &= L_2 - \zeta_1 \partial_{1} - t_1 \partial_{3} + {4 \over 5} \lambda_4 t_3 \partial_{1}, &
\mathcal{L}_6 &= L_6 - \zeta_3 \partial_{3} + {3 \over 5} \lambda_8 t_3 \partial_{1}.
\end{align*}
For $m,l \in \{1,2\}$, $m \leqslant l$, and $k \in \{1,2,3,4,6\}$, the commutation relations are
\begin{align*}
[\mathcal{L}_0, \mathcal{L}_k] &= k \mathcal{L}_k,  & \quad [\mathcal{L}_1, \mathcal{L}_3] &= 0,
\\
[\mathcal{L}_1, \mathcal{L}_2] &= \wp_{1,1} \mathcal{L}_1 - \mathcal{L}_3, & [\mathcal{L}_3, \mathcal{L}_2] &= \left(\wp_{1,3} + {4 \over 5} \lambda_4 \right) \mathcal{L}_1, \\
[\mathcal{L}_1, \mathcal{L}_4] &= \wp_{1,3} \mathcal{L}_1 + \wp_{1,1} \mathcal{L}_3, & [\mathcal{L}_3, \mathcal{L}_4] &= \left(\wp_{3,3} + {6 \over 5} \lambda_6 \right) \mathcal{L}_1 + \left(\wp_{1,3} - \lambda_4\right) \mathcal{L}_3,\\
[\mathcal{L}_1, \mathcal{L}_6] &= \wp_{1,3} \mathcal{L}_3, & [\mathcal{L}_3, \mathcal{L}_6] &= {3 \over 5} \lambda_8 \mathcal{L}_1 + \wp_{3,3} \mathcal{L}_3,\\
[\mathcal{L}_{2m}, \mathcal{L}_{2l+2}] &= \sum_{s=0}^{3} c_{2m,2l+2}^{2s}(\lambda) \mathcal{L}_{2s}
-\frac{1}{2} \wp_{2m-1,2l-1,3} \mathcal{L}_1 +\frac{1}{2} \wp_{1,2m-1,2l-1} \mathcal{L}_3.  \hspace{-250 pt} & &
\end{align*}
\end{thm}

The expressions for the coordinates $(\lambda)$ and $(p)$ are (see eq. (4.3)-(4.6) and (4.8) in~\cite{B2}):
\begin{align*}
\lambda_{4} &= - 3 b_{1,1}^2 + {1 \over 2} b_{3,1} - 2 b_{1,3}, \\
\lambda_{6} &= 2 b_{1,1}^3 + {1 \over 4} b_{2,1}^2 - {1 \over 2}  b_{1,1} b_{3,1} - 2 b_{1,1} b_{1,3} + {1 \over 2} b_{3,3}, \\
\lambda_{8} &= \left(4 b_{1,1}^2 + b_{1,3}\right) b_{1,3} - {1 \over 2} (b_{3,1} b_{1,3} - b_{2,1} b_{2,3} + b_{1,1} b_{3,3}), \hspace{-120pt} \\
\quad \lambda_{10} &= 2 b_{1,1} b_{1,3}^2 + {1 \over 4} b_{2,3}^2 - {1 \over 2} b_{1,3} b_{3,3},& 
p_{3,3} &= 6 b_{1,1} b_{1,3} - b_{3,3}.
\end{align*}

\begin{thm}[Section 4 in~\cite{B2}] \label{t2g} In the case of genus $g=2$
the homogeneous polynomial vector fields $\mathcal{D}_0,\mathcal{D}_1, \mathcal{D}_2, \mathcal{D}_3, \mathcal{D}_4, \mathcal{D}_6$ give a~solution to Problem~\ref{p2}. Here the vector fields~$\mathcal{D}_1, \mathcal{D}_3$ and~$\mathcal{D}_0$ are determined by Lemmas \ref{ld1}, \ref{ld2} and~\ref{ld0} for $g=2$. Set $p_{1,k} = 2 b_{1,k}$ and $q_{i,j,k} = \mathcal{D}_i(p_{j,k})$ for $i,j,k \in \{1,3\}$. The polynomial vector fields~$\mathcal{D}_2, \mathcal{D}_4, \mathcal{D}_6$ are determined by the conditions
\begin{align}
\mathcal{D}_2(b_{1,1}) &= {8 \over 5} \lambda_{4} + 2 b_{1,1}^2 + 4 b_{1,3}, \label{thisf} \\ \mathcal{D}_4(b_{1,1}) &= {2 \over 5} \lambda_{6} - 2 b_{1,1} b_{1,3} + b_{3,3}, \nonumber \\
\mathcal{D}_6(b_{1,1}) &= {1 \over 5} \lambda_{8} + {1 \over 2} \left(b_{3,1} b_{1,3} - b_{1,1} b_{3,3}\right) - b_{1,3}^2, \nonumber \\
\mathcal{D}_2(b_{1,3}) &= - {4 \over 5} \lambda_{4} b_{1,1} + 2 b_{1,1} b_{1,3}, \nonumber
\\
\mathcal{D}_4(b_{1,3}) &= - {6 \over 5} \lambda_{6} b_{1,1} + 2 \lambda_{4} b_{1,3} - 4 b_{1,1}^2 b_{1,3} + b_{3,1} b_{1,3} - {1 \over 2} b_{2,1} b_{2,3}, \nonumber
\\
\mathcal{D}_6(b_{1,3}) &= - {8 \over 5} \lambda_{8} b_{1,1} + 2 \lambda_{6} b_{1,3} - 2 b_{1,1} b_{1,3}^2 - b_{2,3}^2 + b_{1,3} b_{3,3}. \nonumber
\end{align}
and the relations
\begin{align} \label{comd2}
[\mathcal{D}_1, \mathcal{D}_2] &= b_{1,1} \mathcal{D}_1 - \mathcal{D}_3, & [\mathcal{D}_1, \mathcal{D}_4] &= b_{1,3} \mathcal{D}_1 + b_{1,1} \mathcal{D}_3, & [\mathcal{D}_1, \mathcal{D}_6] &= b_{1,3} \mathcal{D}_3.
\end{align}
These relations determine the polynomials $\mathcal{D}_k(b_{i,j})$ for $k \in \{2,4,6\}$, $i \in \{2,3\}$, $j \in \{1,3\}$.
For $m,l \in \{1,2\}$, $m \leqslant l$, and $k \in \{1,2,3,4,6\}$, the commutation relations are \eqref{comd2}~and
\begin{align*}
[\mathcal{D}_0, \mathcal{D}_k] &= k \mathcal{D}_k  & [\mathcal{D}_1, \mathcal{D}_3] &= 0, \\
[\mathcal{D}_3, \mathcal{D}_2] &= \left(b_{1,3} + {4 \over 5} \lambda_4 \right) \mathcal{D}_1, &
\quad [\mathcal{D}_3, \mathcal{D}_4] &= \left({1 \over 2} p_{3,3} + {6 \over 5} \lambda_6 \right) \mathcal{D}_1 + \left(b_{1,3} - \lambda_4\right) \mathcal{D}_3,\\
[\mathcal{D}_3, \mathcal{D}_6] &= {3 \over 5} \lambda_8 \mathcal{D}_1 + {1 \over 2} p_{3,3} \mathcal{D}_3,
\\
[\mathcal{D}_{2m}, \mathcal{D}_{2l+2}] &= \sum_{k=0}^{3} c_{2m,2l+2}^{2k}(\lambda) \mathcal{D}_{2k}
-\frac{1}{4} q_{2m-1,2l-1,3} \mathcal{D}_1 +\frac{1}{4} q_{1,2m-1,2l-1} \mathcal{D}_3. \hspace{-250pt}
\end{align*}
\end{thm}

\begin{proof}
This Theorem follows from Theorem 4.6 in \cite{B2} and the expressions in the proof of this Theorem.
Note that this expressions are based on the relation \eqref{keyrel} for $\mathcal{L}_k$ given in Theorem \ref{tg2}.
See also Theorems 4.7 and B.3 in \cite{B2} and Section 8 in~\cite{B3}.
\end{proof}

Let us note that the result of Theorem \ref{t2g} is equivalent to the corresponding result in~Example 15 in \cite{BL}. Some misprints in \cite{BL} were corrected in \cite{B2}. In terms of the work~\cite{BL}, the polynomial vector fields $\mathcal{D}_k$ determine the action of the operators $\mathcal{L}_k$ on $\mathcal{F}$. The latter is determined by the values $\mathcal{L}_k(\wp_{1,1}) \in \mathcal{F}$, that correspond to $\mathcal{D}_k(b_{1,1})$, given by \eqref{thisf}.

\vfill

\eject

\section{Lie algebra of derivations of hyperelliptic functions: genus~$3$ case} \label{S3L}

\begin{thm}[Theorem 10.1 and Corollary 10.2 in~\cite{B3}] \label{tg3} In the case of genus $g= 3$ the following vector fields give a~solution to Problem \ref{p1}
\begin{align*}
\mathcal{L}_1 &= \partial_{1}, \qquad \mathcal{L}_3 = \partial_{3}, \qquad \mathcal{L}_5 = \partial_{5}, \\
\mathcal{L}_0 &= L_0 - t_1 \partial_{1} - 3 t_3 \partial_{3} - 5 t_5 \partial_{5}, \\
\mathcal{L}_2 &= L_2 - \left(\zeta_1 - {8 \over 7} \lambda_4 t_3 \right) \partial_{1}
- \left( t_1 - {4 \over 7} \lambda_4 t_5 \right) \partial_{3} - 3 t_3 \partial_{5}, \\
\mathcal{L}_4 &= L_4 - \left(\zeta_3 - {12 \over 7} \lambda_6 t_3 \right) \partial_{1}
- \left( \zeta_1 + \lambda_4 t_3 - {6 \over 7} \lambda_6 t_5 \right) \partial_{3} - (t_1 + 3 \lambda_4 t_5) \partial_{5}, \\
\mathcal{L}_6 &= L_6 - \left(\zeta_5 - {9 \over 7} \lambda_8 t_3 \right) \partial_{1}
- \left(\zeta_3 - {8 \over 7} \lambda_8 t_5 \right) \partial_{3}
- \left(\zeta_1 + \lambda_4 t_3 + 2 \lambda_6 t_5 \right) \partial_{5}, \\
\mathcal{L}_8 &= L_8 + \left({6 \over 7} \lambda_{10} t_3 - \lambda_{12} t_5\right) \partial_{1}
- \left(\zeta_5 - {10 \over 7} \lambda_{10} t_5 \right) \partial_{3}
- \left(\zeta_3 + \lambda_8 t_5 \right) \partial_{5}, \\
\mathcal{L}_{10} &= L_{10} + \left( {3 \over 7} \lambda_{12} t_3 - 2 \lambda_{14} t_5 \right) \partial_{1}
+ {5 \over 7} \lambda_{12} t_5 \partial_{3} - \zeta_5 \partial_{5}, 
\end{align*}
The commutation relations are
\begin{align*}
[\mathcal{L}_0, \mathcal{L}_k] &= k \mathcal{L}_k, \quad k=1,2,3,4,5,6,8,10, & [\mathcal{L}_1, \mathcal{L}_3] &= 0, & [\mathcal{L}_1, \mathcal{L}_5] &= 0, & [\mathcal{L}_3, \mathcal{L}_5] &= 0,
\end{align*}
\[
  \begin{pmatrix}
 [\mathcal{L}_1, \mathcal{L}_2] \\
 [\mathcal{L}_1, \mathcal{L}_4] \\
 [\mathcal{L}_1, \mathcal{L}_6] \\
 [\mathcal{L}_1, \mathcal{L}_8] \\
 [\mathcal{L}_1, \mathcal{L}_{10}] 
 \end{pmatrix}
= 
 \begin{pmatrix}
 \wp_{1,1} & -1 & 0\\
 \wp_{1,3} & \wp_{1,1} & -1 \\
 \wp_{1,5} & \wp_{1,3} & \wp_{1,1} \\
 0 & \wp_{1,5} & \wp_{1,3} \\
 0 & 0 & \wp_{1,5}
 \end{pmatrix}
 \begin{pmatrix}
  \mathcal{L}_1 \\ \mathcal{L}_3\\ \mathcal{L}_5
 \end{pmatrix},
\]
\[
  \begin{pmatrix}
 [\mathcal{L}_3, \mathcal{L}_2] \\
 [\mathcal{L}_3, \mathcal{L}_4] \\
 [\mathcal{L}_3, \mathcal{L}_6] \\
 [\mathcal{L}_3, \mathcal{L}_8] \\
 [\mathcal{L}_3, \mathcal{L}_{10}] 
 \end{pmatrix} = 
 {3 \over 7}
 \begin{pmatrix}
 5 \lambda_4 \\
 4 \lambda_6 \\
 3 \lambda_8 \\
 2 \lambda_{10} \\
 \lambda_{12}
 \end{pmatrix}
 \mathcal{L}_1
+
 \begin{pmatrix}
 \wp_{1,3} - \lambda_4 & 0 & - 3\\
 \wp_{3,3} & \wp_{1,3} - \lambda_4 & 0 \\
 \wp_{3,5} & \wp_{3,3} & \wp_{1,3} - \lambda_4 \\
 0 & \wp_{3,5} & \wp_{3,3}\\
 0 & 0 & \wp_{3,5}
 \end{pmatrix}
 \begin{pmatrix}
  \mathcal{L}_1 \\ \mathcal{L}_3\\ \mathcal{L}_5
 \end{pmatrix},
\]
\[
  \begin{pmatrix}
 [\mathcal{L}_5, \mathcal{L}_2] \\
 [\mathcal{L}_5, \mathcal{L}_4] \\
 [\mathcal{L}_5, \mathcal{L}_6] \\
 [\mathcal{L}_5, \mathcal{L}_8] \\
 [\mathcal{L}_5, \mathcal{L}_{10}] 
 \end{pmatrix} = 
 {2 \over 7}
 \begin{pmatrix}
 2 \lambda_{4} \\
 3 \lambda_6 \\
 4 \lambda_8 \\
 5 \lambda_{10} \\
 6 \lambda_{12} 
 \end{pmatrix}
 \mathcal{L}_3
 - \begin{pmatrix}
 0 \\
 3 \lambda_4 \\
 2 \lambda_6 \\
 \lambda_8 \\
 0
 \end{pmatrix}
 \mathcal{L}_5
+ 
 \begin{pmatrix}
 \wp_{1,5} & 0 & 0\\
 \wp_{3,5} & \wp_{1,5} & 0 \\
 \wp_{5,5} & \wp_{3,5} & \wp_{1,5} \\
 - \lambda_{12} & \wp_{5,5} & \wp_{3,5}\\
 -2 \lambda_{14} & - \lambda_{12} & \wp_{5,5}
 \end{pmatrix}
 \begin{pmatrix}
  \mathcal{L}_1 \\ \mathcal{L}_3\\ \mathcal{L}_5
 \end{pmatrix},
\]
\begin{equation} \label{L2k}
\begin{pmatrix}
[\mathcal{L}_2, \mathcal{L}_4]\\
 [\mathcal{L}_2, \mathcal{L}_6]\\
 [\mathcal{L}_2, \mathcal{L}_8]\\
 [\mathcal{L}_2, \mathcal{L}_{10}]
\end{pmatrix} = \begin{pmatrix}
\sum_{k=0}^{5} c_{2,4}^{2k}(\lambda) \mathcal{L}_{2k}\\
\sum_{k=0}^{5} c_{2,6}^{2k}(\lambda) \mathcal{L}_{2k} \\
\sum_{k=0}^{5} c_{2,8}^{2k}(\lambda) \mathcal{L}_{2k}\\
\sum_{k=0}^{5} c_{2,10}^{2k}(\lambda) \mathcal{L}_{2k}\\
\end{pmatrix} + {1 \over 2}
\begin{pmatrix}
 - \wp_{1,1,3} & \wp_{1,1,1} & 0 \\
- \wp_{1,3,3} - \wp_{1,1,5} & \wp_{1,1,3} & \wp_{1,1,1} \\
- 2 \wp_{1,3,5} & \wp_{1,1,5} & \wp_{1,1,3} \\
- \wp_{1,5,5} & 0 & \wp_{1,1,5} \\
\end{pmatrix}  \begin{pmatrix}
 \mathcal{L}_1 \\
 \mathcal{L}_3 \\
 \mathcal{L}_5 \\
 \end{pmatrix},
\end{equation}
\[
 \begin{pmatrix}
 [\mathcal{L}_4, \mathcal{L}_6]\\
 [\mathcal{L}_4, \mathcal{L}_8]\\
 [\mathcal{L}_4, \mathcal{L}_{10}]\\
 [\mathcal{L}_6, \mathcal{L}_8]\\
 [\mathcal{L}_6, \mathcal{L}_{10}]\\
 [\mathcal{L}_8, \mathcal{L}_{10}]
 \end{pmatrix} 
=  \begin{pmatrix}
\sum_{k=0}^{5} c_{4,6}^{2k}(\lambda) \mathcal{L}_{2k}\\
\sum_{k=0}^{5} c_{4,8}^{2k}(\lambda) \mathcal{L}_{2k}\\
\sum_{k=0}^{5} c_{4,10}^{2k}(\lambda) \mathcal{L}_{2k}\\
\sum_{k=0}^{5} c_{6,8}^{2k}(\lambda) \mathcal{L}_{2k}\\
\sum_{k=0}^{5} c_{6,10}^{2k}(\lambda) \mathcal{L}_{2k}\\
\sum_{k=0}^{5} c_{8,10}^{2k}(\lambda) \mathcal{L}_{2k}
 \end{pmatrix} + {1 \over 2}
 \begin{pmatrix}
 - \wp_{3,3,3} & \wp_{1,3,3} - 2 \wp_{1,1,5} & 2 \wp_{1,1,3} \\
 - 2 \wp_{3,3,5} & 0 & 2 \wp_{1,3,3} \\
 - \wp_{3,5,5} & - \wp_{1,5,5} & 2 \wp_{1,3,5} \\
- 2 \wp_{3,5,5} & 2 \wp_{1,5,5} - \wp_{3,3,5} & \wp_{3,3,3} \\
- \wp_{5,5,5} & - \wp_{3,5,5} & \wp_{3,3,5} + \wp_{1,5,5}\\
0 & - \wp_{5,5,5} & \wp_{3,5,5} 
 \end{pmatrix}
 \begin{pmatrix}
 \mathcal{L}_1 \\
 \mathcal{L}_3 \\
 \mathcal{L}_5 \\
 \end{pmatrix}.
\]
\end{thm}

The expressions for the coordinates $(\lambda)$ and $(p)$ are (see Section 9 in~\cite{B3}):
\begin{align*}
\lambda_{4} &= - 3 b_{1,1}^2 + {1 \over 2} b_{3,1} - 2 b_{1,3}, \qquad
\lambda_{6} = 2 b_{1,1}^3 + {1 \over 4} b_{2,1}^2 - {1 \over 2} b_{1,1} b_{3,1} - 2 b_{1,1} b_{1,3} + {1 \over 2} b_{3,3} - 2 b_{1,5}, \\
\lambda_{8} &= 4 b_{1,1}^2 b_{1,3} - {1 \over 2} (b_{3,1} b_{1,3} - b_{2,1} b_{2,3} + b_{1,1} b_{3,3}) + b_{1,3}^2 - 2 b_{1,1} b_{1,5} + {1 \over 2} b_{3,5}, \\
\lambda_{10} &= 2 b_{1,1} b_{1,3}^2 + {1 \over 4} b_{2,3}^2 - {1 \over 2} b_{1,3} b_{3,3} - \frac{1}{2} (b_{3,1} b_{1,5} - b_{2,1} b_{2,5} + b_{1,1} b_{3,5}) + \left(4 b_{1,1}^2 + 2 b_{1,3}\right) b_{1,5}, \\
\lambda_{12} &= 4 b_{1,1} b_{1,3} b_{1,5} - {1 \over 2} (b_{3,3} b_{1,5} - b_{2,3} b_{2,5} + b_{1,3} b_{3,5}) + b_{1,5}^2, \\
\lambda_{14} &= 2 b_{1,1} b_{1,5}^2 + {1 \over 4} b_{2,5}^2 - {1 \over 2} b_{1,5} b_{3,5},
\qquad
p_{3,3} = 6 b_{1,1} b_{1,3} - b_{3,3} + 6 b_{1,5}, \\
p_{3,5} &= 6 b_{1,1} b_{1,5} - b_{3,5}, \qquad 
p_{5,5} = (b_{3,1} b_{1,5} - b_{2,1} b_{2,5} + b_{1,1} b_{3,5}) - 2 \left(4 b_{1,1}^2 + b_{1,3}\right) b_{1,5}.
\end{align*}

\begin{thm} In the case of genus $g=3$
the homogeneous polynomial vector fields $\mathcal{D}_s$ for $s \in \{0, 1, 2, 3, 4, 5, 6, 8, 10\}$ give a~solution to Problem~\ref{p2}. The vector fields~$\mathcal{D}_1, \mathcal{D}_3, \mathcal{D}_5$ and $\mathcal{D}_0$ are determined by Lemmas \ref{ld1}, \ref{ld2} and~\ref{ld0} for~$g=3$.
We set $p_{1,k} = 2 b_{1,k}$ and $q_{i,j,k} = \mathcal{D}_i(p_{j,k})$ for $i,j,k \in \{1,3,5\}$.
The polynomial vector fields~$\mathcal{D}_2$ and $\mathcal{D}_4$ are determined by the conditions 
\begin{align*}
\mathcal{D}_2(b_{1,1}) &= {12 \over 7} \lambda_{4} + 2 b_{1,1}^2 + 4 b_{1,3}, \qquad
\mathcal{D}_4(b_{1,1}) = {4 \over 7} \lambda_6 - 2 b_{1,1} b_{1,3} + b_{3,3} + 2 b_{1,5}, \\
\mathcal{D}_2(b_{1,3}) &= - {8 \over 7} \lambda_4 b_{1,1} + 2 b_{1,1} b_{1,3} + 6 b_{1,5}, \\
\mathcal{D}_4(b_{1,3}) &= - {12 \over 7} \lambda_6 b_{1,1} - 10 b_{1,1}^2 b_{1,3} + 2 b_{3,1} b_{1,3} - 4 b_{1,3}^2 - {1 \over 2} b_{2,1} b_{2,3} + 2 b_{1,1} b_{1,5},\\
\mathcal{D}_2(b_{1,5}) &= - {4 \over 7} \lambda_4 b_{1,3} + 2 b_{1,1} b_{1,5}, \\
\mathcal{D}_4(b_{1,5}) &= - {6 \over 7} \lambda_6 b_{1,3} - 16 b_{1,1}^2 b_{1,5}  + 3 b_{3,1} b_{1,5} - 8 b_{1,3} b_{1,5} - {1 \over 2} b_{2,1} b_{2,5},
\end{align*}
and the relations 
\begin{equation} \label{rel3}
\begin{pmatrix}
 [\mathcal{D}_1, \mathcal{D}_2] \\
 [\mathcal{D}_1, \mathcal{D}_4] \\
 \end{pmatrix}
= 
 \begin{pmatrix}
 b_{1,1} & -1 & 0\\
 b_{1,3} & b_{1,1} & -1 \\
 \end{pmatrix}
 \begin{pmatrix}
  \mathcal{D}_1 \\ \mathcal{D}_3\\ \mathcal{D}_5
 \end{pmatrix}.
\end{equation}

The polynomial vector fields~$\mathcal{D}_6, \mathcal{D}_8$ and $\mathcal{D}_{10}$ are determined by the relations
\begin{align}
\mathcal{D}_{6} &=  {1 \over 4} (2 [\mathcal{D}_2, \mathcal{D}_{4}] + b_{2,3} \mathcal{D}_1 - b_{2,1} \mathcal{D}_3) - {8 \over 7} \left( \lambda_{6} \mathcal{D}_0 - \lambda_4 \mathcal{D}_{2}\right), \label{this} \\
\mathcal{D}_{8} &= {1 \over 8} \left(2 [\mathcal{D}_2, \mathcal{D}_{6}] + \left({1 \over 2} \mathcal{D}_1(p_{3,3}) + b_{2,5}\right) \mathcal{D}_1 - b_{2,3} \mathcal{D}_3 - b_{2,1} \mathcal{D}_5\right) - {3 \over 7} \left( \lambda_{8} \mathcal{D}_0 - \lambda_4 \mathcal{D}_{4}\right), \nonumber
\\
\mathcal{D}_{10} &= {1 \over 12} (2 [\mathcal{D}_2, \mathcal{D}_{8}] + \mathcal{D}_1(p_{3,5}) \mathcal{D}_1 - b_{2,5} \mathcal{D}_3 - b_{2,3} \mathcal{D}_5) - {4 \over 21} \left( \lambda_{10} \mathcal{D}_0 - \lambda_4 \mathcal{D}_{6}\right). \nonumber
\end{align}

The commutation relations are~obtained from the relations in Theorem \ref{tg3} by the correspondence
\begin{align*}
\mathcal{L}_k &= \varphi^* \mathcal{D}_k, & \varphi^* p_{i,j} &= 2 \wp_{i,j} & \varphi^* q_{i,j,k} &= 2 \wp_{i,j,k}.
\end{align*}
\end{thm}

\begin{proof}
This Theorem follows from Section 9 in~\cite{B3}. Note that the relations \eqref{rel3} determine the polynomials $\mathcal{D}_k(b_{i,j})$ for $k \in \{2,4\}$, $i \in \{2,3\}$, $j \in \{1,3,5\}$.
For $\mathcal{L}_k$ given in Theorem~\ref{tg3} the relation \eqref{keyrel} holds.
Cf. eq. \eqref{3ex}, \eqref{L2k}, and \eqref{this}.
\end{proof}

The polynomial dynamical systems in $\mathbb{C}^9$ determined
by the vector fields $\mathcal{D}_k$ are presented in Section 8 of \cite{BB21}.

\vfill

\eject

\section{Lie algebra of derivations of hyperelliptic functions: genus~$4$ case} \label{S4L}

\begin{ex}\label{ex61} In the case of genus $g=4$ Lemma~\ref{lcijk} implies
\begin{align}\label{C2}
\begin{pmatrix}
[L_2, L_4]\\
[L_2, L_6]\\
[L_2, L_8]\\
[L_2, L_{10}]\\
[L_2, L_{12}]\\
[L_2, L_{14}]
\end{pmatrix} &= 
 {1 \over 9} \begin{pmatrix}
24 \lambda_6 & - 24 \lambda_4 & 0 & 18 & 0 & 0 & 0 & 0\\
20 \lambda_8 & 0 & - 20 \lambda_4 & 0 & 36 & 0 & 0 & 0\\
16 \lambda_{10} & 0 & 0 & - 16 \lambda_4 & 0 & 54 & 0 & 0\\
12 \lambda_{12} & 0 & 0 & 0 & - 12 \lambda_4 & 0 & 72 & 0\\
8 \lambda_{14} & 0 & 0 & 0 & 0 & - 8 \lambda_4 & 0 & 90\\ 
4 \lambda_{10}& 0 & 0 & 0 & 0 & 0 & - 4 \lambda_4 & 0
\end{pmatrix}
\widetilde{L},
\end{align}
where $\widetilde{L} = \begin{pmatrix}
L_0&
L_2&
L_4&
L_6&
L_8&
L_{10}&
L_{12}&
L_{14}
\end{pmatrix}^\top.$ 
\end{ex}
We denote the matrix $(c_{2,2j}^{2k}(\lambda))$ by
 $\mathcal{C}_2(\lambda)$. The right hand side of \eqref{C2} is~$\mathcal{C}_2(\lambda)\widetilde{L}$. 
 
\begin{thm}[cf. Corollary 6.2 in \cite{BB20}] 
\label{tg4} In the case of genus $g=4$ the following vector fields give a~solution to Problem \ref{p1}
\begin{align*}
\mathcal{L}_1 &= \partial_1, \qquad \mathcal{L}_3 = \partial_3,
\qquad \mathcal{L}_5 = \partial_5,
\qquad \mathcal{L}_7 = \partial_7,\\
\mathcal{L}_0 &= L_0
- t_1 \partial_1 - 3 t_3 \partial_3 - 5 t_5 \partial_5 - 7 t_7 \partial_7,
\\
\mathcal{L}_2 &= L_2 
- \zeta_1 \partial_1 + {4 \over 3} \lambda_4 t_{3} \partial_1 - \left(t_1 - {8 \over 9} \lambda_4 t_{5}\right) \partial_3 - \left(3 t_3 - {4 \over 9} \lambda_4 t_{7}\right) \partial_5 - 5 t_5 \partial_7,
\\
\mathcal{L}_4 &= L_4
- \zeta_3 \partial_1 - \zeta_1 \partial_3 + \\ & \quad
+ 2 \lambda_6 t_3 \partial_1
- \left( \lambda_4 t_3
- {4 \over 3} \lambda_6 t_5 \right) \partial_3 
- \left( t_1 + 3 \lambda_4 t_5 
- {2 \over 3} \lambda_6 t_7\right) \partial_5
- \left( 3 t_3 + 5 \lambda_4 t_7 \right) \partial_7,
\\
\mathcal{L}_6 &= L_6
- \zeta_5 \partial_1 - \zeta_3 \partial_3 - \zeta_1 \partial_5 + \\ & \quad
+ {5 \over 3} \lambda_8 t_3 \partial_1 + {16 \over 9} \lambda_8 t_5 \partial_3 - \left( \lambda_4 t_3 + 2 \lambda_6 t_5 -{8 \over 9} \lambda_8 t_7 \right) \partial_5 - \left(t_1 +3 \lambda_4 t_5 + 4 \lambda_6 t_7\right) \partial_7,
\\
\mathcal{L}_8 &= L_8
- \zeta_7 \partial_1 - \zeta_5 \partial_3 - \zeta_3 \partial_5 - \zeta_1 \partial_7 + \left({4 \over 3} \lambda_{10} t_3 - \lambda_{12} t_5\right) \partial_1 + \\ & \quad + {20 \over 9} \lambda_{10} t_5 \partial_3 - \left(\lambda_8 t_5 - {10 \over 9} \lambda_{10} t_7\right) \partial_5 - \left(\lambda_4 t_3 + 2 \lambda_6 t_5 + 3 \lambda_8 t_7\right) \partial_7,
\\
\mathcal{L}_{10} &= L_{10} 
- \zeta_7 \partial_3 - \zeta_5 \partial_5 - \zeta_3 \partial_7  + \\ & \quad
+ (\lambda_{12} t_3 - 2 \lambda_{14} t_5 - \lambda_{16} t_7) \partial_1 + {5 \over 3} \lambda_{12} t_5 \partial_3 + {4 \over 3} \lambda_{12} t_7 \partial_5 - (\lambda_8 t_5 + 2 \lambda_{10} t_7) \partial_7,
\\
\mathcal{L}_{12} &= L_{12}
- \zeta_7 \partial_5 - \zeta_5 \partial_7 +\\ & \quad  + \left({2 \over 3} \lambda_{14} t_3 - 3 \lambda_{16} t_5 - 2 \lambda_{18} t_7\right) \partial_1 + \left({10 \over 9} \lambda_{14} t_5 - \lambda_{16} t_7 \right) \partial_3 + {14 \over 9} \lambda_{14} t_7 \partial_5 - \lambda_{12} t_7 \partial_7,
\\
\mathcal{L}_{14} &= L_{14}
- \zeta_7 \partial_7 + \left({1 \over 3} \lambda_{16} t_3 - 4 \lambda_{18} t_5\right) \partial_1 + \left({5 \over 9} \lambda_{16} t_5 - 2 \lambda_{18} t_7\right) \partial_3 + {7 \over 9} \lambda_{16} t_7 \partial_5.
\end{align*}
We denote this Lie algebra by $\mathscr{L}$. The commutation relations are given in Lemma \ref{lem01}, Lemma \ref{l23}, Lemma \ref{l24}, and Corollary \ref{corlast}.
\end{thm}

\begin{lem} \label{lem01}
For the commutators in the Lie algebra $\mathscr{L}$ we have the relations:
\begin{align*}
[\mathcal{L}_0, \mathcal{L}_{k}] &= k \mathcal{L}_{k}, & &k = 0,1,2,3,4,5,6,7,8,10,12,14;\\
[\mathcal{L}_{k}, \mathcal{L}_{m}] &= 0, & &k,m = 1,3,5,7.
\end{align*}
\end{lem}
\begin{proof}
We use the explicit expressions for $\mathcal{L}_k$ and the fact that $\mathcal{L}_0$ is the Euler vector field, thus $\mathcal{L}_0 \zeta_k = k \zeta_k$.
\end{proof}

\vfill
 
\begin{lem} \label{l23}
For the commutators in the Lie algebra $\mathscr{L}$ we have the relations:
\begin{align*}
\begin{pmatrix}
[\mathcal{L}_1, \mathcal{L}_2]\\
[\mathcal{L}_1, \mathcal{L}_4]\\
[\mathcal{L}_1, \mathcal{L}_6]\\
[\mathcal{L}_1, \mathcal{L}_8]\\
[\mathcal{L}_1, \mathcal{L}_{10}]\\ 
[\mathcal{L}_1, \mathcal{L}_{12}]\\
[\mathcal{L}_1, \mathcal{L}_{14}] 
\end{pmatrix}
&=
\begin{pmatrix}
\wp_{1,1} & - 1 & 0 & 0\\
\wp_{1,3} &\wp_{1,1} & - 1 & 0\\
\wp_{1,5} & \wp_{1,3} &\wp_{1,1} & - 1\\
\wp_{1,7} & \wp_{1,5} & \wp_{1,3} &\wp_{1,1}\\
0 & \wp_{1,7} & \wp_{1,5} & \wp_{1,3}\\
0 & 0 & \wp_{1,7} & \wp_{1,5}\\
0 & 0 & 0 & \wp_{1,7}\\
\end{pmatrix}
\begin{pmatrix}
\mathcal{L}_1\\
\mathcal{L}_3\\
\mathcal{L}_5\\
\mathcal{L}_7
\end{pmatrix},
\\
\begin{pmatrix}
[\mathcal{L}_3, \mathcal{L}_2]\\
[\mathcal{L}_3, \mathcal{L}_4]\\
[\mathcal{L}_3, \mathcal{L}_6]\\
[\mathcal{L}_3, \mathcal{L}_8]\\
[\mathcal{L}_3, \mathcal{L}_{10}]\\ 
[\mathcal{L}_3, \mathcal{L}_{12}]\\
[\mathcal{L}_3, \mathcal{L}_{14}] 
\end{pmatrix}
&=
\begin{pmatrix}
\wp_{1,3} - \lambda_4 & 0 & - 3 & 0\\
\wp_{3,3} & \wp_{1,3} - \lambda_4 & 0 & - 3\\
\wp_{3,5} & \wp_{3,3} & \wp_{1,3} - \lambda_4 & 0\\
\wp_{3,7} & \wp_{3,5} & \wp_{3,3} & \wp_{1,3} - \lambda_4\\
0 & \wp_{3,7} & \wp_{3,5} &  \wp_{3,3}\\
0 & 0 & \wp_{3,7} & \wp_{3,5}\\
0 & 0 & 0 & \wp_{3,7}
\end{pmatrix}
\begin{pmatrix}
\mathcal{L}_1\\
\mathcal{L}_3\\
\mathcal{L}_5\\
\mathcal{L}_7
\end{pmatrix}
+ {1 \over 3} 
\begin{pmatrix}
 7 \lambda_4\\
 6 \lambda_6\\
 5 \lambda_8\\
 4 \lambda_{10}\\
 3 \lambda_{12}\\
 2 \lambda_{14}\\
 \lambda_{16}\\
\end{pmatrix}
\mathcal{L}_1,
\\
\begin{pmatrix}
[\mathcal{L}_5, \mathcal{L}_2]\\
[\mathcal{L}_5, \mathcal{L}_4]\\
[\mathcal{L}_5, \mathcal{L}_6]\\
[\mathcal{L}_5, \mathcal{L}_8]\\
[\mathcal{L}_5, \mathcal{L}_{10}]\\ 
[\mathcal{L}_5, \mathcal{L}_{12}]\\
[\mathcal{L}_5, \mathcal{L}_{14}] 
\end{pmatrix}
&=
\begin{pmatrix}
\wp_{1,5} & 0 & 0 & - 5\\
\wp_{3,5} &\wp_{1,5} & - 3 \lambda_4 & 0\\
\wp_{5,5} & \wp_{3,5} & \wp_{1,5} - 2 \lambda_6 & - 3 \lambda_4\\
\wp_{5,7} - \lambda_{12} &  \wp_{5,5} & \wp_{3,5} - \lambda_{8} &\wp_{1,5} - 2 \lambda_6\\
- 2 \lambda_{14} & \wp_{5,7} - \lambda_{12} & \wp_{5,5} & \wp_{3,5} - \lambda_{8}\\
- 3 \lambda_{16} & - 2 \lambda_{14} & \wp_{5,7} &  \wp_{5,5}\\
- 4 \lambda_{18} & - 3 \lambda_{16} & 0 & \wp_{5,7} \\
\end{pmatrix}
\begin{pmatrix}
\mathcal{L}_1\\
\mathcal{L}_3\\
\mathcal{L}_5\\
\mathcal{L}_7
\end{pmatrix}
 + {4 \over 9} 
\begin{pmatrix}
 2 \lambda_4\\
 3 \lambda_6\\
 4 \lambda_8\\
 5 \lambda_{10}\\
 6 \lambda_{12}\\
 7 \lambda_{14}\\
 8 \lambda_{16}\\
\end{pmatrix}
\mathcal{L}_3,\\
\begin{pmatrix} 
[\mathcal{L}_7, \mathcal{L}_2]\\
[\mathcal{L}_7, \mathcal{L}_4]\\
[\mathcal{L}_7, \mathcal{L}_6]\\
[\mathcal{L}_7, \mathcal{L}_8]\\
[\mathcal{L}_7, \mathcal{L}_{10}]\\ 
[\mathcal{L}_7, \mathcal{L}_{12}]\\
[\mathcal{L}_7, \mathcal{L}_{14}] 
\end{pmatrix}
&=
\begin{pmatrix}
\wp_{1,7} & 0 & 0 & 0\\
\wp_{3,7} &\wp_{1,7} &  0 & - 5 \lambda_4\\
\wp_{5,7} & \wp_{3,7} & \wp_{1,7} & - 4 \lambda_6 \\
\wp_{7,7} & \wp_{5,7} &  \wp_{3,7} &\wp_{1,7} - 3 \lambda_8\\
- \lambda_{16} & \wp_{7,7}&  \wp_{5,7} & \wp_{3,7} - 2 \lambda_{10}\\
- 2 \lambda_{18} & -\lambda_{16} & \wp_{7,7} &  \wp_{5,7} - \lambda_{12}\\
0 & - 2 \lambda_{18} & - \lambda_{16} & \wp_{7,7} \\
\end{pmatrix}
\begin{pmatrix}
\mathcal{L}_1\\
\mathcal{L}_3\\
\mathcal{L}_5\\
\mathcal{L}_7
\end{pmatrix}
 + {2 \over 9} 
\begin{pmatrix}
 2 \lambda_4\\
 3 \lambda_6\\
 4 \lambda_8\\
 5 \lambda_{10}\\
 6 \lambda_{12}\\
 7 \lambda_{14}\\
 8 \lambda_{16}\\
\end{pmatrix}
\mathcal{L}_5.
\end{align*}
\end{lem}

\textit{The proof} follows from the explicit expressions for $\mathcal{L}_k$.

\begin{lem} \label{l24}
For the commutators in the Lie algebra $\mathscr{L}$ we have the relations
\begin{align*}
\begin{pmatrix}
[\mathcal{L}_2, \mathcal{L}_4]\\
[\mathcal{L}_2, \mathcal{L}_6]\\
[\mathcal{L}_2, \mathcal{L}_8]\\
[\mathcal{L}_2, \mathcal{L}_{10}]\\ 
[\mathcal{L}_2, \mathcal{L}_{12}]\\
[\mathcal{L}_2, \mathcal{L}_{14}] 
\end{pmatrix} & =
 \mathcal{C}_2(\lambda)
\begin{pmatrix}
\mathcal{L}_0\\
\mathcal{L}_2\\
\mathcal{L}_4\\
\mathcal{L}_6\\
\mathcal{L}_8\\
\mathcal{L}_{10}\\ 
\mathcal{L}_{12}\\
\mathcal{L}_{14}
\end{pmatrix} + {1 \over 2}
\begin{pmatrix}
- \wp_{1,1,3} & \wp_{1,1,1} & 0 & 0 \\
- \wp_{1,3,3} - \wp_{1,1,5} & \wp_{1,1,3} & \wp_{1,1,1} & 0 \\
- 2 \wp_{1,3,5} - \wp_{1,1,7} & \wp_{1,1,5} & \wp_{1,1,3} & \wp_{1,1,1} \\
- 2 \wp_{1,3,7} - \wp_{1,5,5} & \wp_{1,1,7} & \wp_{1,1,5} & \wp_{1,1,3} \\
- 2 \wp_{1,5,7} & 0 & \wp_{1,1,7} & \wp_{1,1,5} \\
- \wp_{1,7,7} & 0 & 0 & \wp_{1,1,7} 
\end{pmatrix}
\begin{pmatrix}
\mathcal{L}_1\\
\mathcal{L}_3\\
\mathcal{L}_5\\
\mathcal{L}_7
\end{pmatrix},
\end{align*}
where the polynomial matrix $\mathcal{C}_2(\lambda) = (c_{2,2j}^{2k}(\lambda))$ is given in \eqref{C2}.
\end{lem}

\begin{proof}
From Theorem 6.3 in \cite{BB20} we obtain the expressions for $\mathcal{L}_{2k} \zeta_s$, where~$s = 1,3,5,7$ and $k = 1,2,3,4,5,6,7$.
Provided this, the proof follows from the explicit expressions for~$\mathcal{L}_k$.
\end{proof}

\section{Generators in the polynomial Lie algebra in $\mathbb{C}^{12}$} \label{s10}

In the case of genus $g=4$ the expressions for the coordinates $(\lambda)$ and $(p)$ in $b_{i,j}$ obtained in Corollaries \ref{cor3} and \ref{cor4} are the following:
\begin{align} 
\lambda_{4} &= - 3 b_{1,1}^2 + \frac{1}{2} b_{3,1} - 2 b_{1, 3}, \qquad
\lambda_{6} = 2 b_{1,1}^3 + \frac{1}{4} b_{2,1}^2 - \frac{1}{2} b_{1,1} b_{3,1} - 2 b_{1,1} b_{1, 3} + \frac{1}{2} b_{3, 3} - 2 b_{1, 5},  \nonumber\\
\lambda_{8} &= 4 b_{1,1}^2 b_{1, 3} + b_{1, 3} b_{1, 3} - 2 b_{1,1} b_{1, 5} + \frac{1}{2} b_{3, 5} - 2 b_{1, 7} - \frac{1}{2} (b_{3,1} b_{1, 3} - b_{2,1} b_{2, 3} + b_{1,1} b_{3, 3}),\label{lam4} \\
\lambda_{10} &= 2 (b_{1,1} b_{1, 3}^2 + b_{1, 5} b_{1, 3} - b_{1,1} b_{1, 7} + 2 b_{1,1}^2 b_{1, 5}) + {1 \over 4} b_{2,3}^2 - \frac{1}{2} (b_{1, 3} b_{3, 3} - b_{3, 7} + b_{3,1} b_{1, 5} - b_{2,1} b_{2, 5} + b_{1,1} b_{3, 5} ), \nonumber\\
\lambda_{12} &= 4 b_{1,1} (b_{1, 3} b_{1, 5} + b_{1,1} b_{1, 7}) + b_{1, 5}^2 + 2 b_{1, 3} b_{1, 7} - \nonumber \\
& \qquad - \frac{1}{2} \left( b_{3, 3} b_{1, 5} - b_{2,3} b_{2,5}
+ b_{1, 3}b_{3, 5} + b_{3,1} b_{1, 7} - b_{2,1} b_{2, 7} + b_{1,1} b_{3, 7}\right), \nonumber \\
\lambda_{14} &= 2 b_{1,1} b_{1, 5}^2 + {1 \over 4} b_{2,5}^2 - \frac{1}{2} b_{1, 5} b_{3, 5} +  4 b_{1,1} b_{1,3} b_{1,7} + 2 b_{1,5} b_{1,7} - \frac{1}{2} \left( b_{3,3} b_{1,7} - b_{2,3} b_{2,7} + b_{1,3} b_{3,7} \right), \nonumber \\
\lambda_{16} &= 4 b_{1,1} b_{1, 5} b_{1, 7} + b_{1,7}^2 - \frac{1}{2} \left(b_{3, 5} b_{1,7} - b_{2,5} b_{2,7} +  b_{1,5} b_{3, 7}\right), \qquad
\lambda_{18} = 2 b_{1,1} b_{1, 7}^2 + {1 \over 4} b_{2,7}^2 - \frac{1}{2} b_{1,7} b_{3, 7},
\nonumber \\
p_{3, 3} &= 6 b_{1,1} b_{1, 3} - b_{3, 3} + 6 b_{1, 5}, \qquad
p_{3, 5} = 6 b_{1,1} b_{1, 5} - b_{3, 5} + 6 b_{1, 7}, \qquad
p_{3, 7} = 6 b_{1,1} b_{1, 7} - b_{3, 7},
\nonumber \\
p_{5, 5} &= b_{3,1} b_{1, 5} - b_{2,1} b_{2, 5} + b_{1,1} b_{3, 5} - 2 \left(4 b_{1,1}^2 + b_{1, 3}\right) b_{1, 5} + 10 b_{1,1} b_{1, 7} - 2 b_{3, 7}, \nonumber\\
p_{5, 7} &= b_{3,1} b_{1,7} - b_{2,1} b_{2,7} + b_{1,1} b_{3,7} - 2 \left(4 b_{1,1}^2 + b_{1, 3}\right) b_{1, 7}, \nonumber \\
p_{7,7} &= b_{3,3} b_{1,7} - b_{2,3} b_{2,7} + b_{1,3} b_{3,7} - 2 \left( 4 b_{1,1} b_{1,3} + b_{1,5}\right) b_{1,7}. \nonumber
\end{align}

\begin{thm} \label{td4} In the case of genus $g=4$
the homogeneous polynomial vector fields~$\mathcal{D}_s$ for $s \in \{0,1,2,3,4,5,6,7,8,10,12,14\}$ give a~solution to Problem~\ref{p2}. The vector fields $\mathcal{D}_1, \mathcal{D}_3, \mathcal{D}_5, \mathcal{D}_7$, and $\mathcal{D}_0$ are determined by Lemmas \ref{ld1}, \ref{ld2} and~\ref{ld0} for~$g=4$. We set~$p_{i,j}' = \mathcal{D}_1(p_{i,j})$ and $p_{i,j}'' = \mathcal{D}_1(\mathcal{D}_1(p_{i,j}))$ for $i,j \in \{3,5,7\}$. The polynomial vector fields~$\mathcal{D}_2$ and~$\mathcal{D}_4$ are determined by the conditions 
\begin{align}
\mathcal{D}_2(b_{1,1}) &= - b_{1,1}^2 + {1 \over 2} b_{3,1} + 2 b_{1,3} + {7 \over 9} \lambda_4,
\label{con1} \\
\mathcal{D}_2(b_{1,3}) &= - b_{1,1} b_{1,3} + {1 \over 2} b_{3,3} + 3 b_{1,5} - {4 \over 3} \lambda_4 b_{1,1} + {1 \over 2} p_{3,3}, \nonumber \\
\mathcal{D}_2(b_{1,5}) &= - b_{1,1} b_{1,5} + {1 \over 2} b_{3,5} + 5 b_{1,7} - {8 \over 9} \lambda_4 b_{1,3} + {1 \over 2} p_{3,5}, \nonumber \\
\mathcal{D}_2(b_{1,7}) &= - b_{1,1} b_{1,7} + {1 \over 2} b_{3,7} - {4 \over 9} \lambda_4 b_{1,5} + {1 \over 2} p_{3,7}, \nonumber
\\
\mathcal{D}_4(b_{1,1}) &= - 2 b_{1,1} b_{1,3}  + b_{3,3} + 2 b_{1,5} + {2 \over 3} \lambda_6, \nonumber \\
\mathcal{D}_4(b_{1,3}) &= - b_{1,3}^2 + 3 b_{1,7} + \lambda_4 b_{1,3} - 2 \lambda_6 b_{1,1} - \lambda_8 - {1 \over 2} b_{1,1} p_{3,3} + {1 \over 2} p_{3,3}'' + {1 \over 2} p_{3,5}, \nonumber \\
\mathcal{D}_4(b_{1,5}) &= - b_{1,3} b_{1,5} + 3 \lambda_4 b_{1,5} - {4 \over 3} \lambda_6 b_{1,3} - {1 \over 2} b_{1,1} p_{3,5} + {1 \over 2} p_{3,5}'' + {1 \over 2} p_{5,5}, \nonumber \\
\mathcal{D}_4(b_{1,7}) &= - b_{1,3} b_{1,7} + 5 \lambda_4 b_{1,7} - {2 \over 3} \lambda_6 b_{1,5} - {1 \over 2} b_{1,1} p_{3,7} + {1 \over 2} p_{3,7}'' + {1 \over 2} p_{5,7}, \nonumber
\end{align}
and the relations
\begin{equation} \label{con2}
\begin{pmatrix} 
[\mathcal{D}_1, \mathcal{D}_2]\\
[\mathcal{D}_1, \mathcal{D}_4]\\
\end{pmatrix}
=
\begin{pmatrix}
b_{1,1} & - 1 & 0\\
b_{1,3} & b_{1,1} & - 1\\
\end{pmatrix}
\begin{pmatrix}
\mathcal{D}_1\\
\mathcal{D}_3\\
\mathcal{D}_5
\end{pmatrix}.
\end{equation}

The polynomial vector fields~$\mathcal{D}_6, \mathcal{D}_8$ and $\mathcal{D}_{10}$ are determined by the relations
\begin{align}
\mathcal{D}_6 &= {1 \over 4} \left(2 [\mathcal{D}_2, \mathcal{D}_4] + b_{2, 3} \mathcal{D}_1 - b_{2,1} \mathcal{D}_3 \right) - {4 \over 3 } \left( \lambda_{6} \mathcal{D}_0 - \lambda_4 \mathcal{D}_2\right), \nonumber\\
\mathcal{D}_8 &= {1 \over 8} \left(2 [\mathcal{D}_2, \mathcal{D}_6] + \left({1 \over 2} p_{3, 3}' + b_{2,5}\right) \mathcal{D}_1 - b_{2,3} \mathcal{D}_3 - b_{2,1} \mathcal{D}_5\right) - {5 \over 9 } \left( \lambda_{8} \mathcal{D}_0 - \lambda_4 \mathcal{D}_4\right), \label{con3}
\\
\mathcal{D}_{10} &= {1 \over 12} \left(2 [\mathcal{D}_2, \mathcal{D}_8] + \left(p_{3,5}'+ b_{2,7}\right) \mathcal{D}_1 - b_{2,5} \mathcal{D}_3 - b_{2,3} \mathcal{D}_5 - b_{2,1} \mathcal{D}_7\right) - {8 \over 27 } \left( \lambda_{10} \mathcal{D}_0 - \lambda_4 \mathcal{D}_6\right), \nonumber
\\
\mathcal{D}_{12} &= {1 \over 16} \left(2 [\mathcal{D}_2, \mathcal{D}_{10}] + \left(p_{3,7}' + {1 \over 2} p_{5,5}'\right) \mathcal{D}_1 - b_{2,7} \mathcal{D}_3 - b_{2,5} \mathcal{D}_5 - b_{2,3} \mathcal{D}_7\right) - { 1 \over 6 } \left( \lambda_{12} \mathcal{D}_0 - \lambda_4 \mathcal{D}_8\right), \nonumber \\
\mathcal{D}_{14} &= {1 \over 20} \left(2 [\mathcal{D}_2, \mathcal{D}_{12}] + p_{5, 7}' \mathcal{D}_1 - b_{2,7} \mathcal{D}_5 - b_{2,5} \mathcal{D}_7 \right) - {4 \over 45 } \left( \lambda_{14} \mathcal{D}_0 - \lambda_4 \mathcal{D}_{10}\right). \nonumber
\end{align}
We denote  by $\mathscr{D}$ the polynomial Lie algebra generated by the polynomial vector fields $\mathcal{D}_s$ for $s \in \{0,1,2,3,4,5,6,7,8,10,12,14\}$. The commutation relations for $\mathscr{D}$ are given in~Lemma~\ref{c41}, Lemma~\ref{c42}, Lemma \ref{c43}, and Lemma \ref{c44}. 
\end{thm}

\begin{proof}
First we note that the polynomial vector fields $\mathcal{D}_s$ are determined uniquely by the~Theorem.
The polynomial vector fields $\mathcal{D}_1, \mathcal{D}_3, \mathcal{D}_5, \mathcal{D}_7$, and $\mathcal{D}_0$ are determined by Lemmas~\ref{ld1},~\ref{ld2} and~\ref{ld0}.
Relations \eqref{con1} determine the polynomials $\mathcal{D}_2(b_{1,j})$ and $\mathcal{D}_4(b_{1,j})$ for $j \in \{1,3,5,7\}$.
Relations \eqref{con2} determine the polynomials $\mathcal{D}_2(b_{i,j})$ and $\mathcal{D}_4(b_{i,j})$ for $i \in \{2,3\}$, $j \in \{1,3,5,7\}$. Relations \eqref{con3} determine the polynomial vector fields $\mathcal{D}_{2k}$ for $k \in \{3,4,5,6,7\}$ given the polynomial vector fields $\mathcal{D}_{s}$ for $s \in \{1,3,5,7\}$ and $\mathcal{D}_{2m}$ for~$m<k$.

Thus we have a system of $12$ polynomial vector fields that are explicitly defined in the coordinates $(b)$. By Corollary \ref{cords},
the polynomial vector fields $\mathcal{D}_s$ for $s \in \{1,3,5,7\}$ are projectable for the polynomial map $\rho\colon \mathbb{C}^{12} \to \mathbb{C}^{8}$. Their pushforwards are zero.
By Corollary \ref{cor0}, the polynomial vector field $\mathcal{D}_0$ is projectable for the polynomial map~$\rho$ with pushforward $L_0$. It is a direct calculation to check that the polynomial vector fields $\mathcal{D}_{2k}$ for $k \in \{1,2,3,4,5,6,7\}$ are projectable for the polynomial map~$\rho$ with pushforwards~$L_{2k}$. This proves the claim of the Theorem that the vector fields~$\mathcal{D}_s$ give a~solution to Problem~\ref{p2}.

The commutation relations given in Section \ref{scom4} in~Lemmas~\ref{c41},~\ref{c42}, ~\ref{c43}, and~\ref{c44} are proved by direct computation.
\end{proof}

\begin{thm} \label{tcomp}
For the polynomial vector fields $\mathcal{D}_k$ for $k \in \{0,1,2,3,4,5,6,7,8,10,12,14\}$, determined by Theorem \ref{td4}, we have
\begin{equation} \label{refres}
\mathcal{L}_k (\varphi^* b_{i,j}) = \varphi^* \mathcal{D}_k (b_{i,j}), 
\end{equation}
where the vector fields $\mathcal{L}_k$ are given in Theorem \ref{tg4}. 
\end{thm}

\begin{proof}
For the polynomial vector fields $\mathcal{D}_1, \mathcal{D}_3, \mathcal{D}_5, \mathcal{D}_7$, and $\mathcal{D}_0$ the claim of the Theorem is given by Lemmas \ref{ld1}, \ref{ld2}, and \ref{ld0}. For $k \in \{1,2,3,4,5,6,7\}$ define the vector fields~$\widetilde{\mathcal{D}}_{2s}$ by the relation \eqref{refres}. The vector fields $\widetilde{\mathcal{D}}_{2s}$ are linear combinations of the fields $\mathcal{D}_k$ for $k \in \{0,1,2,3,4,5,6,7,8,10,12,14\}$ with coefficients rational functions in $(b)$. The pushforward of $\widetilde{\mathcal{D}}_{2s}$ is $L_{2s}$, which is equal to the pushforward of $\mathcal{D}_{2s}$. Therefore we have
\[
\widetilde{\mathcal{D}}_{2s} = \mathcal{D}_{2s} + \alpha_{2s}^1(b) \mathcal{D}_{1} + \alpha_{2s}^3(b) \mathcal{D}_{3} + \alpha_{2s}^5(b) \mathcal{D}_{5} + \alpha_{2s}^7(b) \mathcal{D}_{7} 
\]
for some rational functions $\alpha_{2s}^j(b)$. Comparing  results of Lemma \ref{l23} and Lemma \ref{c42}, we~obtain $\mathcal{D}_{i}(\alpha_{2s}^j(b)) = 0$ for $i \in \{1,3,5,7\}$, thus $\alpha_{2s}^j(b)$ is a rational function in $(\lambda)$. We~have $\wt \alpha_{2s}^j(b) = 2s - j$, where $j$ is odd, while $\wt \lambda_m$ is even for all $m$ (see \eqref{lam4}), thus we obtain $\alpha_{2s}^j(b) = 0$ and $\widetilde{\mathcal{D}}_{2s} = \mathcal{D}_{2s}$.
\end{proof}

\vfill
\eject

\section{Commutation relations in the polynomial Lie algebra in $\mathbb{C}^{12}$} \label{scom4}

\begin{lem} \label{c41}
The relations hold for the commutators in the polynomial Lie algebra~$\mathscr{D}$:
\begin{align*}
[\mathcal{D}_0, \mathcal{D}_{k}] &= k \mathcal{D}_{k}, & &k = 0,1,2,3,4,5,6,7,8,10,12,14;\\
[\mathcal{D}_{k}, \mathcal{D}_{m}] &= 0, & &k,m = 1,3,5,7.
\end{align*}
\end{lem}

\begin{lem} \label{c42}
The relations hold for the commutators in the polynomial Lie algebra~$\mathscr{D}$:
\begin{align}
\begin{pmatrix}
[\mathcal{D}_1, \mathcal{D}_2]\\
[\mathcal{D}_1, \mathcal{D}_4]\\
[\mathcal{D}_1, \mathcal{D}_6]\\
[\mathcal{D}_1, \mathcal{D}_8]\\
[\mathcal{D}_1, \mathcal{D}_{10}]\\ 
[\mathcal{D}_1, \mathcal{D}_{12}]\\
[\mathcal{D}_1, \mathcal{D}_{14}] 
\end{pmatrix}
&=
\begin{pmatrix}
b_{1,1} & - 1 & 0 & 0\\
b_{1,3} &b_{1,1} & - 1 & 0\\
b_{1,5} & b_{1,3} & b_{1,1} & - 1\\
b_{1,7} & b_{1,5} & b_{1,3} & b_{1,1}\\
0 & b_{1,7} & b_{1,5} & b_{1,3}\\
0 & 0 & b_{1,7} & b_{1,5}\\
0 & 0 & 0 & b_{1,7}\\
\end{pmatrix}
\begin{pmatrix}
\mathcal{D}_1\\
\mathcal{D}_3\\
\mathcal{D}_5\\
\mathcal{D}_7
\end{pmatrix}, \label{d14}
\\
\begin{pmatrix}
[\mathcal{D}_3, \mathcal{D}_2]\\
[\mathcal{D}_3, \mathcal{D}_4]\\
[\mathcal{D}_3, \mathcal{D}_6]\\
[\mathcal{D}_3, \mathcal{D}_8]\\
[\mathcal{D}_3, \mathcal{D}_{10}]\\ 
[\mathcal{D}_3, \mathcal{D}_{12}]\\
[\mathcal{D}_3, \mathcal{D}_{14}] 
\end{pmatrix}
&=
{1 \over 2}
\begin{pmatrix}
2 b_{1,3} - 2 \lambda_4 & 0 & - 6 & 0\\
p_{3,3} & 2 b_{1,3} - 2 \lambda_4 & 0 & - 6\\
p_{3,5} & p_{3,3} & 2 b_{1,3} - 2 \lambda_4 & 0\\
p_{3,7} & p_{3,5} & p_{3,3} & 2 b_{1,3} - 2 \lambda_4\\
0 & p_{3,7} & p_{3,5} & p_{3,3}\\
0 & 0 & p_{3,7} & p_{3,5}\\
0 & 0 & 0 & p_{3,7}
\end{pmatrix}
\begin{pmatrix}
\mathcal{D}_1\\
\mathcal{D}_3\\
\mathcal{D}_5\\
\mathcal{D}_7
\end{pmatrix}
+ {1 \over 3} 
\begin{pmatrix}
 7 \lambda_4\\
 6 \lambda_6\\
 5 \lambda_8\\
 4 \lambda_{10}\\
 3 \lambda_{12}\\
 2 \lambda_{14}\\
 \lambda_{16}\\
\end{pmatrix}
\mathcal{D}_1, \nonumber
\\
\begin{pmatrix}
[\mathcal{D}_5, \mathcal{D}_2]\\
[\mathcal{D}_5, \mathcal{D}_4]\\
[\mathcal{D}_5, \mathcal{D}_6]\\
[\mathcal{D}_5, \mathcal{D}_8]\\
[\mathcal{D}_5, \mathcal{D}_{10}]\\ 
[\mathcal{D}_5, \mathcal{D}_{12}]\\
[\mathcal{D}_5, \mathcal{D}_{14}] 
\end{pmatrix}
&=
{1 \over 2}
\begin{pmatrix}
2 b_{1,5} & 0 & 0 & - 10\\
p_{3,5} & 2 b_{1,5} & - 6 \lambda_4 & 0\\
p_{5,5} & p_{3,5} & 2 b_{1,5} - 4 \lambda_6 & - 6 \lambda_4\\
p_{5,7} - 2 \lambda_{12} & p_{5,5} & p_{3,5} - 2 \lambda_{8} & 2 b_{1,5} - 4 \lambda_6\\
- 4 \lambda_{14} & p_{5,7} - 2 \lambda_{12} & p_{5,5} & p_{3,5} - 2 \lambda_{8}\\
- 6 \lambda_{16} & - 4 \lambda_{14} & p_{5,7} & p_{5,5}\\
- 8 \lambda_{18} & - 6 \lambda_{16} & 0 & p_{5,7} \\
\end{pmatrix}
\begin{pmatrix}
\mathcal{D}_1\\
\mathcal{D}_3\\
\mathcal{D}_5\\
\mathcal{D}_7
\end{pmatrix}
 + {4 \over 9} 
\begin{pmatrix}
 2 \lambda_4\\
 3 \lambda_6\\
 4 \lambda_8\\
 5 \lambda_{10}\\
 6 \lambda_{12}\\
 7 \lambda_{14}\\
 8 \lambda_{16}\\
\end{pmatrix}
\mathcal{D}_3, \nonumber \\
\begin{pmatrix} 
[\mathcal{D}_7, \mathcal{D}_2]\\
[\mathcal{D}_7, \mathcal{D}_4]\\
[\mathcal{D}_7, \mathcal{D}_6]\\
[\mathcal{D}_7, \mathcal{D}_8]\\
[\mathcal{D}_7, \mathcal{D}_{10}]\\ 
[\mathcal{D}_7, \mathcal{D}_{12}]\\
[\mathcal{D}_7, \mathcal{D}_{14}] 
\end{pmatrix}
&=
{1 \over 2}
\begin{pmatrix}
2 b_{1,7} & 0 & 0 & 0\\
p_{3,7} & 2 b_{1,7} &  0 & - 10 \lambda_4\\
p_{5,7} & p_{3,7} & 2 b_{1,7} & - 8 \lambda_6 \\
p_{7,7} & p_{5,7} & p_{3,7} & 2 b_{1,7} - 6 \lambda_8\\
- 2 \lambda_{16} & p_{7,7}& p_{5,7} & p_{3,7} - 4 \lambda_{10}\\
- 4 \lambda_{18} & - 2 \lambda_{16} & p_{7,7} & p_{5,7} - 2 \lambda_{12}\\
0 & - 4 \lambda_{18} & - 2 \lambda_{16} & p_{7,7} \\
\end{pmatrix}
\begin{pmatrix}
\mathcal{D}_1\\
\mathcal{D}_3\\
\mathcal{D}_5\\
\mathcal{D}_7
\end{pmatrix}
 + {2 \over 9} 
\begin{pmatrix}
 2 \lambda_4\\
 3 \lambda_6\\
 4 \lambda_8\\
 5 \lambda_{10}\\
 6 \lambda_{12}\\
 7 \lambda_{14}\\
 8 \lambda_{16}\\
\end{pmatrix}
\mathcal{D}_5. \nonumber
\end{align}
\end{lem}

\begin{lem} \label{c43}
The relations hold for the commutators in the polynomial Lie algebra~$\mathscr{D}$:
\begin{align*}
\begin{pmatrix}
[\mathcal{D}_2, \mathcal{D}_4]\\
[\mathcal{D}_2, \mathcal{D}_6]\\
[\mathcal{D}_2, \mathcal{D}_8]\\
[\mathcal{D}_2, \mathcal{D}_{10}]\\ 
[\mathcal{D}_2, \mathcal{D}_{12}]\\
[\mathcal{D}_2, \mathcal{D}_{14}] 
\end{pmatrix} & =
 \mathcal{C}_2(\lambda)
\begin{pmatrix}
\mathcal{D}_0\\
\mathcal{D}_2\\
\mathcal{D}_4\\
\mathcal{D}_6\\
\mathcal{D}_8\\
\mathcal{D}_{10}\\ 
\mathcal{D}_{12}\\
\mathcal{D}_{14}
\end{pmatrix} + {1 \over 2}
\begin{pmatrix}
- b_{2,3} & b_{2,1} & 0 & 0 \\
- {1 \over 2} p_{3,3}' - b_{2,5} & b_{2,3} & b_{2,1} & 0 \\
- p_{3,5}' - b_{2,7} & b_{2,5} & b_{2,3} & b_{2,1} \\
- p_{3,7}' - {1 \over 2} p_{5,5}' & b_{2,7} & b_{2,5} & b_{2,3} \\
- p_{5,7}' & 0 & b_{2,7} & b_{2,5} \\
- {1 \over 2} p_{7,7}' & 0 & 0 & b_{2,7} 
\end{pmatrix}
\begin{pmatrix}
\mathcal{D}_1\\
\mathcal{D}_3\\
\mathcal{D}_5\\
\mathcal{D}_7
\end{pmatrix},
\end{align*}
where the polynomial matrix $\mathcal{C}_2(\lambda) = (c_{2,2j}^{2k}(\lambda))$ is given in \eqref{C2}.
\end{lem}

\begin{lem} \label{c44}
Set $p_{1,k} = 2 b_{1,k}$ and $q_{i,j,k} = \mathcal{D}_i(p_{j,k})$ for $i,j,k \in \{1,3,5,7\}$.
The relations for $m,l \in {2,3,4,5,6}$, $m \leqslant l$ hold for the commutators in the polynomial Lie algebra~$\mathscr{D}$:
\[
 [\mathcal{D}_{2m}, \mathcal{D}_{2l+2}] = \sum_{k=0}^{7} c_{2m,2l+2}^{2k}(\lambda) \mathcal{D}_{2k} + (A_{2m,2l+2}) \begin{pmatrix}
\mathcal{D}_1 &
\mathcal{D}_3 &
\mathcal{D}_5 &
\mathcal{D}_7
\end{pmatrix}^\top,
\]
where
\[
\begin{pmatrix}
A_{4,6}\\
A_{4,8}\\
A_{4,10}\\ 
A_{4,12}\\
A_{4,14}\\
A_{6,8}\\
A_{6,10}\\ 
A_{6,12}\\
A_{6,14}\\
A_{8,10}\\ 
A_{8,12}\\
A_{8,14}\\
A_{10,12}\\ 
A_{10,14}\\
A_{12,14}
\end{pmatrix} = {1 \over 4}
\begin{pmatrix}
- q_{3,3,3} & q_{1,3,3} - 2 q_{1,1,5} & 2 q_{1,1,3} & 0 \\
- 2 q_{3,3,5} & - 2 q_{1,1,7} & 2 q_{1,3,3} &  2 q_{1,1,3} \\
- 2 q_{3,3,7} - q_{3,5,5} & - q_{1,5,5} & 2 q_{1,3,5} & 2 q_{1,3,3} \\
- 2 q_{3,5,7} & - 2 q_{1,5,7} & 2 q_{1,3,7} & 2 q_{1,3,5} \\
- q_{3,7,7} & - q_{1,7,7} & 0 & 2 q_{1,3,7} \\
q_{3,3,7} - 2 q_{3,5,5} & 2 q_{1,5,5} -  q_{3,3,5} - 2 q_{1,3,7} & q_{3,3,3} - 2 q_{1,1,7} & q_{1,3,3} + 2 q_{1,1,5}\\
- q_{5,5,5} - 2 q_{3,5,7} & 2 q_{1,5,7} - q_{3,5,5} - q_{3,3,7} & q_{3,3,5} +  q_{1,5,5} - 2 q_{1,3,7} & q_{3,3,3} + 2 q_{1,3,5} \\
- 2 q_{5,5,7} & - 2 q_{3,5,7} & q_{3,3,7} & q_{3,3,5} + 2 q_{1,5,5} \\
- q_{5,7,7} & - q_{3,7,7} & -  q_{1,7,7} & q_{3,3,7} + 2 q_{1,5,7} \\
- 2 q_{3,7,7} - q_{5,5,7} & - q_{5,5,5} + 2 q_{1,7,7} & 2 q_{1,5,7} + q_{3,5,5} - 2 q_{3,3,7} & 2 q_{3,3,5} - q_{1,5,5} \\
- 2 q_{5,7,7} & - 2 q_{5,5,7} & 2 q_{1,7,7} & 2 q_{3,5,5} \\
- q_{7,7,7} & - q_{5,7,7} & - q_{3,7,7} & 2 q_{3,5,7} + q_{1,7,7}\\
0 & - 2 q_{5,7,7} & 2 q_{3,7,7} - q_{5,5,7}& q_{5,5,5} \\
0 & - q_{7,7,7} & - q_{5,7,7} & q_{5,5,7} + q_{3,7,7} \\
0 & 0 & - q_{7,7,7} & q_{5,7,7} \\
\end{pmatrix}.
\]
\end{lem}

\begin{cor}[From Lemma \ref{c44} and Theorem \ref{tcomp}] \label{corlast}
The relations for $m,l \in {2,3,4,5,6}$, $m \leqslant l$ hold for the commutators in the polynomial Lie algebra~$\mathscr{L}$:
\[
 [\mathcal{L}_{2m}, \mathcal{L}_{2l+2}] = \sum_{k=0}^{7} c_{2m,2l+2}^{2k}(\lambda) \mathcal{L}_{2k} + (\mathcal{A}_{2m,2l+2}) \begin{pmatrix}
\mathcal{L}_1 &
\mathcal{L}_3 &
\mathcal{L}_5 &
\mathcal{L}_7
\end{pmatrix}^\top,
\]
where
\[
\begin{pmatrix}
\mathcal{A}_{4,6}\\
\mathcal{A}_{4,8}\\
\mathcal{A}_{4,10}\\ 
\mathcal{A}_{4,12}\\
\mathcal{A}_{4,14}\\
\mathcal{A}_{6,8}\\
\mathcal{A}_{6,10}\\ 
\mathcal{A}_{6,12}\\
\mathcal{A}_{6,14}\\
\mathcal{A}_{8,10}\\ 
\mathcal{A}_{8,12}\\
\mathcal{A}_{8,14}\\
\mathcal{A}_{10,12}\\ 
\mathcal{A}_{10,14}\\
\mathcal{A}_{12,14}
\end{pmatrix} = {1 \over 2}
\begin{pmatrix}
- \wp_{3,3,3} & \wp_{1,3,3} - 2 \wp_{1,1,5} & 2 \wp_{1,1,3} & 0 \\
- 2 \wp_{3,3,5} & - 2 \wp_{1,1,7} & 2 \wp_{1,3,3} &  2 \wp_{1,1,3} \\
- 2 \wp_{3,3,7} - \wp_{3,5,5} & - \wp_{1,5,5} & 2 \wp_{1,3,5} & 2 \wp_{1,3,3} \\
- 2 \wp_{3,5,7} & - 2 \wp_{1,5,7} & 2 \wp_{1,3,7} & 2 \wp_{1,3,5} \\
- \wp_{3,7,7} & - \wp_{1,7,7} & 0 & 2 \wp_{1,3,7} \\
\wp_{3,3,7} - 2 \wp_{3,5,5} & 2 \wp_{1,5,5} -  \wp_{3,3,5} - 2 \wp_{1,3,7} & \wp_{3,3,3} - 2 \wp_{1,1,7} & \wp_{1,3,3} + 2 \wp_{1,1,5}\\
- \wp_{5,5,5} - 2 \wp_{3,5,7} & 2 \wp_{1,5,7} - \wp_{3,5,5} - \wp_{3,3,7} & \wp_{3,3,5} +  \wp_{1,5,5} - 2 \wp_{1,3,7} & \wp_{3,3,3} + 2 \wp_{1,3,5} \\
- 2 \wp_{5,5,7} & - 2 \wp_{3,5,7} & \wp_{3,3,7} & \wp_{3,3,5} + 2 \wp_{1,5,5} \\
- \wp_{5,7,7} & - \wp_{3,7,7} & -  \wp_{1,7,7} & \wp_{3,3,7} + 2 \wp_{1,5,7} \\
- 2 \wp_{3,7,7} - \wp_{5,5,7} & - \wp_{5,5,5} + 2 \wp_{1,7,7} & 2 \wp_{1,5,7} + \wp_{3,5,5} - 2 \wp_{3,3,7} & 2 \wp_{3,3,5} - \wp_{1,5,5} \\
- 2 \wp_{5,7,7} & - 2 \wp_{5,5,7} & 2 \wp_{1,7,7} & 2 \wp_{3,5,5} \\
- \wp_{7,7,7} & - \wp_{5,7,7} & - \wp_{3,7,7} & 2 \wp_{3,5,7} + \wp_{1,7,7}\\
0 & - 2 \wp_{5,7,7} & 2 \wp_{3,7,7} - \wp_{5,5,7}& \wp_{5,5,5} \\
0 & - \wp_{7,7,7} & - \wp_{5,7,7} & \wp_{5,5,7} + \wp_{3,7,7} \\
0 & 0 & - \wp_{7,7,7} & \wp_{5,7,7} \\
\end{pmatrix}.
\]
\end{cor}

\vfill
\eject

\section{Polynomial dynamical systems in $\mathbb{C}^{12}$ related to~differentiations of~genus~$4$ hyperelliptic~functions} \label{S4Sys}

As we have noted throughout the work, the generators $\mathcal{L}_k$ of $\Der \mathcal{F}$ and the polynomial vector fields $\mathcal{D}_k$ in $\mathbb{C}^{3g}$ are~related by \eqref{keyrel}. The polynomial vector fields $\mathcal{D}_0, \mathcal{D}_1$, and $\mathcal{D}_s$ for $s \in \{3, 5, \ldots, 2g - 1\}$ are given is Section \ref{s4}.
Section \ref{Sdyn} gives the 
graded homogeneous polynomial dynamical systems $S_0, S_1$, and $S_s$ for $s \in \{3, 5, \ldots, 2g - 1\}$ in $\mathbb{C}^{3g}$ determined by these vector fields.

In this Section in the case of genus $g=4$ we determine the remaing 
graded homogeneous polynomial dynamical systems $S_{2k}$ for $k \in \{1,2,3,4,5,6,7\}$. By~definition, the dynamical system $S_{2k}$ corresponding to the vector field $\mathcal{D}_{2k}$ is given by
\[
{\partial \over \partial \tau_{2k}} b_{i,j} = \mathcal{D}_{2k}(b_{i,j}).
\]

It has been noted in \cite{BL} in the case of genus $g=2$ that to determine such systems it is sufficient to determine the polynomials $\mathcal{D}_{2k}(b_{1,1})$. Indeed, we have the relation
\[
\mathcal{D}_{2k}(\lambda_m) = L_{2k}(\lambda_m) 
\]
for $m \in \{4,6,8,10,12,14,16,18\}$ that determines the action of $\mathcal{D}_{2k}$ on the coordinates~$(\lambda)$. The relations \eqref{lam4} imply
\begin{align*} 
b_{1, 3} &= - {3\over 2} b_{1,1}^2 + \frac{1}{4} b_{3,1} - {1 \over 2} \lambda_{4}, \\
b_{1, 5} &= b_{1,1}^3 + \frac{1}{8} b_{2,1}^2 - \frac{1}{4} b_{1,1} b_{3,1} - b_{1,1} b_{1, 3} + \frac{1}{4} b_{3, 3} - {1 \over 2} \lambda_{6},  \\
b_{1, 7} &= 2 b_{1,1}^2 b_{1, 3} + {1 \over 2} b_{1, 3} b_{1, 3} - b_{1,1} b_{1, 5} + \frac{1}{4} b_{3, 5} - \frac{1}{4} (b_{3,1} b_{1, 3} - b_{2,1} b_{2, 3} + b_{1,1} b_{3, 3}) - {1 \over 2} \lambda_{8}.
\end{align*}
Corollary \ref{cords} implies $\mathcal{D}_1(\lambda_m) = 0$ and
Lemma \ref{ld1} implies
\begin{align*} 
b_{2, s} &= \mathcal{D}_1(b_{1, s}), & b_{3, s} &= \mathcal{D}_1(\mathcal{D}_1(b_{1, s})), & \text{for } & s \in \{1,3,5,7\}. 
\end{align*}
Finally, the relation \eqref{d14} determines $\mathcal{D}_{2k}(\mathcal{D}_1^{(n)}(b_{1, 1}))$ provided $\mathcal{D}_{2k}(b_{1, 1})$ is given. To conclude, we give explicitly the values $\mathcal{D}_{2k}(b_{1,1})$ for $k \in \{1,2,3,4,5,6,7\}$. By \eqref{con1} we have
\begin{align*}
\mathcal{D}_2(b_{1,1}) &= - b_{1,1}^2  + 2 b_{1,3} + {1 \over 2} b_{3,1} + {7 \over 9} \lambda_4,
\\
\mathcal{D}_4(b_{1,1}) &= b_{3,3} - 2 b_{1,1} b_{1,3} + 2 b_{1,5} + {2 \over 3} \lambda_6.
\end{align*}
The remaining values are
\begin{align*}
\mathcal{D}_6(b_{1,1}) &= -b_{1,3}^2  + 2 b_{1,7} - 2 b_{1,1} b_{1,5} + {1 \over 2} b_{3,1} b_{1, 3} - {1 \over 2} b_{3,3} b_{1,1} + {3 \over 2} b_{3,5} + {5 \over 9} \lambda_8,
\\
\mathcal{D}_8(b_{1,1}) &= b_{1,5} b_{3,1} - b_{1,1} b_{3,5} - 2 b_{1,3} b_{1,5} -2 b_{1,1} b_{1, 7} + 2 b_{3,7} + {4 \over 9} \lambda_{10},
\\
\mathcal{D}_{10}(b_{1,1}) &= -b_{1, 5}^2 - 2 b_{1,3} b_{1,7} + {1 \over 2} b_{3,3} b_{1,5} - {1 \over 2} b_{3,5} b_{1,3} + {3 \over 2} b_{3, 1} b_{1,7} - {3 \over 2} b_{3,7} b_{1, 1} + {1 \over 3} \lambda_{12},
\\
\mathcal{D}_{12}(b_{1,1}) &= b_{1,7} b_{3,3} - b_{1,3} b_{3, 7} - 2 b_{1,5} b_{1,7} + {2 \over 9} \lambda_{14},
\\
\mathcal{D}_{14}(b_{1,1}) &= -b_{1,7}^2 + {1 \over 2} b_{3, 5} b_{1,7} -{1 \over 2} b_{3,7} b_{1,5} + {1 \over 9} \lambda_{16}.
\end{align*}

\end{document}